\documentclass[authoryear, preprint, 3p]{elsarticle}

\usepackage[english]{babel}
\usepackage[utf8]{inputenc}
\usepackage{amsthm}
\usepackage{amsmath}
\usepackage{amssymb}
\usepackage{hyperref}
\usepackage{enumerate}
\usepackage{colortbl}
\usepackage[table]{xcolor}
\usepackage{color,soul}
\usepackage{tikz}
\usepackage{placeins}
\usepackage{subfigure} 
\usepackage{booktabs}
\usepackage{bm}
\usepackage{todonotes}
\usepackage{csquotes}
\usepackage{lscape}
\usepackage{longtable}
\usepackage{natbib}
\usepackage{hyperref}
\usepackage{float}
\biboptions{comma}

\newcommand{\argmax}{\operatorname{argmax}}

\newcommand{\law}{\operatorname{Law}}

\newcommand{\Rr}{{\mathbb{R}}}

\newcommand{\Ff}{{\mathcal{F}}}

\newcommand{\Ee}{{\mathbb{E}}}

\newcommand{\fp}{\mathfrak{p}}
\newcommand{\TI}{\mbox{TI}}
\newcommand{\OFI}{\mbox{OFI}}

\definecolor{cadmiumgreen}{rgb}{0.0, 0.42, 0.24}

\newcommand{\commY}[1]{{\color{cadmiumgreen} #1}}

\usepackage[authormarkup=none]{changes}
\definechangesauthor[name={Final}, color=blue]{F}

\journal{TBA}

\newtheorem{theorem}{Theorem}[section]
\newtheorem{corollary}[theorem]{Corollary}
\newtheorem{lemma}[theorem]{Lemma}
\newtheorem{proposition}[theorem]{Proposition}
\newtheorem{definition}[theorem]{Definition}
\newtheorem{remark}[theorem]{Remark}
\newtheorem{example}[theorem]{Example}
\numberwithin{equation}{section}

\begin{document}
	
\begin{frontmatter}
	
\title{Price formation in financial markets: a game-theoretic perspective}

\author[FGV]{David Evangelista}
\ead{david.evangelista@fgv.br}

\author[FGV]{Yuri Saporito}
\ead{yuri.saporito@gmail.com}

\author[UFF]{Yuri Thamsten}
\ead{ythamsten@id.uff.br}

\address[FGV]{Escola de Matem\'atica Aplicada (EMAp), Funda\c c\~ao Getulio Vargas, Rio de Janeiro, Brasil}
\address[UFF]{Instituto de Matem\'atica e Estat\'istica (IME), Universidade Federal Fluminense, Niter\'{o}i, Brasil}

	
\date{}
	
	
\begin{abstract}
We propose two novel frameworks to study the price formation of an asset negotiated in an order book. Specifically, we develop a game-theoretic model in many-person games and mean-field games, considering costs stemming from limited liquidity. We derive analytical formulas for the formed price in terms of the realized order flow. We also identify appropriate conditions that ensure the convergence of the price we find in the finite population game to that of its mean-field counterpart. We numerically assess our results with a large experiment using high-frequency data from ten stocks listed in the NASDAQ, a stock listed in B3 in Brazil, and a cryptocurrency listed in Binance. 
\end{abstract}	

\begin{keyword}
Price Formation; Optimal Trading; Mean-field Games; Finite Population Games.
\MSC[2010] 
\end{keyword}

\end{frontmatter}	


\section{Introduction}\label{sec:Introduction}

Price formation of asset prices is a process constituted as a result of the interaction of traders through the limit order book (LOB), cf. \cite{easley1995market}. More precisely, traders decide on how to act according to the current state of the book, and this process results in the price of the asset.  Following \cite{lehalle2013market}, we put it in distinction to the \textit{price discovery} perspective, in which we assume that the true asset price exists, for fundamental reasons, and the agents inject their views on the fair price in the order book; assuming market efficiency, we would expect that the long-run ``true'' price is eventually \textit{discovered} in this way. Following the terminology of that work, in the price formation problem, the construction of prices occurs in a forward way, whereas in price discovery frameworks, the price is unraveled in a backward mechanism.

From a fundamental perspective, the study of price formation aims to investigate questions such as the ones \cite{kyle1985continuous} poses.\footnote{``How quickly is new private information about the underlying value of a speculative commodity incorporated into market prices? How valuable is private information to an insider? How does noise trading affect the volatility of prices? What determines the liquidity of a speculative market?'' \cite[Introduction]{kyle1985continuous}} The development of numerous models resulted from these investigations, including those considering liquidity issues, of which we cite \cite{mike2008empirical,daniels2003quantitative,weber2005order,rocsu2009dynamic,cont2010stochastic,avellaneda2008high,citeulike:12386824,lehalle2011high}. In this connection, we refer to the survey on agent-based models \cite{chakraborti2011econophysics}; see also \cite{schinckus2012methodological,mastromatteo2014agent}.  \footnote{The flourishing of the statistical investigation of order books has led to many results,  especially from researchers originally from the physics community, see \cite{bouchaud2002statistical,potters2003more,chakraborti2011econophysics1,lan2007statistical}. A popular designation of the latter approach is \textit{econophysics} --- it represents a direction that has its foundations in analogies between financial markets and physics.}

Furthermore, there is a vast literature taking the price discovery perspective. For instance, we mention the efforts \cite{cohen1981transaction,garbade1979structural,ho1980dealer,ho1981optimal,ho1983dynamics} regarding market microstructure properties in such a framework; see also the survey \cite{biais2005market}, and the references therein. Key econometric measures for price discovery analysis are treated in \cite{hasbrouck1995one,gonzalo1995estimation,gonzalo2001systematic}, and further advances include \cite{booth1999price,chu1999price,dias2016price}. Following the line of \cite{hasbrouck1995one}, in \cite{grammig2005internationally,lien2009new,de2010price,dias2018price,fernandes2018price} these works investigate assets traded in many venues envisaging to identify which one leads the process of impound of novel information.  

In financial practice, knowledge of price formation is quite valuable. In the work of \cite{kyle1985continuous}, the equilibrium configuration implied by the market clearing condition indicates that prices are formed proportionally to the order flow. Under our model, we might interpret the constant of proportionality --- commonly designated as Kyle's lambda --- as the book depth. This aspect is a relevant contribution linking the investigation of the causal effect of the order flow on the price to the study of \textit{price impact}, i.e., understanding how the order flow perturbs the dynamics of the corresponding asset price: ``(...) the correlation between an incoming order (to buy or to sell) and the subsequent price change'', as put by \cite[Section 1]{bouchaud2009price}.


Recently, electricity markets became a strong source of motivation for research in price formation. In effect, we cite \cite{shrivats2020mean}, which considers Solar Renewable Energy Certificate (SREC) markets. The agents decide on how to optimally behave according to their SREC generation and their trading activities. From a market clearing condition, the determination of the price occurs as a result of the optimal behavior of the players. The construction of the performance criteria of the agents relies on their views on their future inventory, as is typical in price formation settings. In \cite{shrivats2020optimal}, there is a numerical assessment of the results of \cite{shrivats2020mean}. Using similar models, to a certain extent, there are various results regarding equilibrium prices in the papers \cite{fujii2020mean,fujii2020finite}, both for infinite and finite populations. In the latter framework, these authors also consider multiple populations in \cite{fujii2019probabilistic}. We also mention \cite{feron2020price}, which treats price formation in intraday electricity markets under intermittent renewable generation; in \cite{feron2020leader} they consider this problem in the presence of a leader.

Understanding the price formation is a matter of high relevance on its own. For instance, we refer to the efforts \cite{brennan1995investment,brennan1996market,locke2007order} in the econometrics literature. Conceptually speaking,  the authors seek to explain price dynamics in terms of order flow. For such a study in the context of cryptocurrencies, see \cite{ciaian2016economics}. 

Moreover, in many of the works we mentioned above, the formed price from the model does replicate many empirical features of the market.  In general, we can say that any plausible explanation of the asset price from the price formation viewpoint yields a benchmark trade execution price. This explanation has applications in Transaction Cost Analysis (TCA). Indeed, given a trading rate, its effect on the market will be to perturb the realized order flow. From our formed price, in terms of the perturbed order flow, we can derive a benchmark for the price impact induced by the strategy under scrutiny.

To take liquidity issues into account, we build upon the phenomenological framework of \cite{bertsimas1998optimal} and\cite{almgren2001optimal}. There, price impacts are twofold, viz., they are either permanent or temporary. This distinction is relative to their persistence in time. Thus, the impact on the dynamics of the price is what we regard as permanent, whereas the impact incurred by a single agent as a consequence of having to consume some layers of the book to have their order completely executed is the temporary one.  We call \textit{order book resilience}. the underlying assumption that the book replenishes immediately from the temporary impact.  Further advances in optimal execution problems generalized the idea of temporary price impact to a transient type of this friction, see \cite{gatheral2012transient}. Also, Obizhaeva and Wang developed a model of a block-shaped book with finite resilience, see \cite{obizhaeva2013optimal}. In \cite{kallsen2014high},  the authors see the Almgren and Chriss model as a high resilience limit of that of Obizhaeva and Wang.

Although the final goal of the above models was to study optimal execution problems, they also serve as a basis for the analysis of many other problems, such as optimal hedging \cite{almgren2016option,bank2017hedging,gueant2017option}, management of accelerated share repurchase contracts \cite{gueant2015accelerated}, and even price formation \cite{feron2020price,feron2020leader}. Here, we will develop our contributions to the price formation problem with a related model. We will consider a temporary price impact in a standard way, but we will not assume a specific shape for the price impact on the asset price dynamics.  In turn, we are willing to establish an explanation for the asset price in terms of its realized order flow.

At first, we will make a simplifying assumption that there are infinitely many players.  In this manner, we will use the machinery of Mean-Field Game (MFG) models. MFGs comprise game-theoretic models serving to investigate equilibria of large and competitive populations of strategically behaving individuals. On many occasions, they are useful precisely because they provide analytic tractability. They were introduced in the mathematical community independently by J-M. Lasry and P-L. Lions, see \cite{ll1,ll2,ll3}, and by M. Huang, P. Caines and R. Malham\'e, see  \cite{Caines2,Caines1}. In a second moment, we  investigate the many-person game corresponding to the MFG one.  We work on a framework of asymmetric information among decentralized homogeneous agents, allowing each agent to impact the price in a rather general way. Our goal is to interpret the results for the infinite population as limiting cases of some for the finite population counterpart. In the literature, the work \cite{gjerstad1998price} is another game-theoretic approach to the price formation problem, but in a context of double auctions.

In the MFG setting, we introduce the notion of competitive equilibria for our problem.  Through this definition, we can analyze the interplay between the supply/demand resulting from the order flow and the corresponding asset prices. We assume continuity of the price and order flow, as well as that they are deterministic, working in an MFG framework akin to that of \cite{gomes2020mean} (cf. \cite{ashrafyan2021duality} and \cite{gomes2021random} for further developments on theoretical aspects). Although our assumptions are strong at this point, we will see that our results are satisfactory empirically. Moreover, our analysis of the finite population counterpart will shed light upon our explanation of the asset price endogenously resulting from our MFG model. 

We resort to many-person games to assess sufficient conditions to ensure the convergence of the price we find as the population size grows to infinity.  We allow for more flexibility in this setting,  such as a heterogeneous informational structure as in \cite{evangelista2020finite}. For analytic tractability, as well as for a clearer comparison to the previous framework, we keep the assumption that the parameters are homogeneous among the players. However, we consider a broader class of admissible prices when defining the competitive equilibria for the finite population games. Thus, we investigate some consequences of considering price impacts stemming from each individual agent -- a ``microstructural impact'', or ``micro-impact'', for short. In terms of this micro-impact, alongside a further assumption on the limiting behavior of the uncertainty in the traders' inventories, we are able to formulate a limiting theorem, thus recovering the formula for the price we found for the MFG.

We organize the remainder of this paper as follows. In Section \ref{sec:MFG}, we analyze competitive equilibria in the MFG setting. Then, in Section \ref{sec:FinitePop}, we investigate the many-person game to gain further insight into the results we obtained for the infinite population framework. We conclude the paper numerically assessing our results with high-frequency data from NASDAQ and other two different financial markets.  In Section \ref{sec:Empirical} We empirically assess the results we derived in the previous sections.  Our many-person approach indicates a regression covariate shown to be relevant empirically.  Finally, in Section \ref{sec:conclusions} we state our general findings and provide our concluding remarks.


\section{Price formation with a MFG model} \label{sec:MFG}
 
\subsection{The model}

\paragraph{\textbf{The market}} We consider a single venue in which infinitely many agents trade a single asset. By this, we mean that their trading in other assets do not influence their decisions in the specific one we are focusing on. We denote its reference price \footnote{We assume that the reference price is given by the asset's midprice, i.e.,  the average between the best bid and best ask quotes.}  by $p$. Our point of view is that $p$ is the result of the interaction among the population through the order book. As a simplifying assumption, we stipulate that $p$ is deterministic, as in \cite{gomes2020mean}. In Section \ref{sec:FinitePop} we will relax this restriction in the finite population setting, and we will provide an appropriate interpretation for the price we find in the present section as a limit when the number of agents goes to infinity.

\paragraph{\textbf{Information structure}} We fix a complete filtered probability space $(\Omega, \Ff, \mathbb{P}, \mathbb{F} = \left\{ \mathcal{F}_t\right\}_{0\leqslant t \leqslant T}),$ with $\mathcal{F} = \mathcal{F}_T,$ supporting a one-dimensional Brownian motion $W.$ Therefore, we currently suppose symmetric information throughout the population.

\paragraph{\textbf{The traders}} Let us consider a representative trader beginning a trading program at time $t \in \left[0,T\right],$ targeting to execute a trade of size $|q|$ shares of the asset, where $q \in \mathbb{R}.$ If $q> 0$ (resp., $q<0$) the agent seeks to liquidate (resp., execute) $q$ shares over $\left[t,T\right].$ The trader controls her turnover rate $\left\{ \nu_u \right\}_{t \leqslant u \leqslant T} \in \mathcal{U}_t,$ where
$$
\mathcal{U}_t := \left\{ \nu = \left\{ \nu_u \right\}_{t\leqslant u \leqslant T} : \nu \text{ is } \mathbb{F}-\text{progressively measurable, and } \mathbb{E}\left[\int_t^T  \nu_u^2 \,du\right] < \infty \right\}
$$
is the set of admissible controls. In this manner, her inventory process $\left\{ q_u \right\}_{t \leqslant u \leqslant T}$ satisfies
\begin{equation} \label{eq:Inventory}
    \begin{cases}
        dq_u = \nu_u\,du + \sigma dW_u,\\
        q_t = q.
    \end{cases}
\end{equation}
Above, we allowed uncertainty in the traders' inventory. As put by \cite{duffie2017size}, this is justified because it represents an effect observed in reality. If the trader is a broker, she will receive an order flow from her clients, something which we can model as \eqref{eq:Inventory}, cf. \cite[Subsection 2.1]{carmona2015probabilistic}. Another explanation for introducing uncertainty in the dynamics \eqref{eq:Inventory} of the inventory is that when the agent actually conducts her trades, at the so-called tactical layer of an execution algorithm, in the terminology of \cite{gueant2016financial}, she is subject to execution uncertainty stemming from operational risks. 

As the trader executes a volume $\nu_u\,du $ over the time interval $\left[u,u+du\right[ $ she must pay or receive --- depending on the sign of $\nu_u$--- the quantity $-p(u)\nu_u\,du.$ If we assume a linear temporary price impact, then the indirect cost per share she incurs from this effect is $-\kappa \nu_u;$ hence,  this trade's total cost is  $- \kappa \nu_u^2\,du.$ Gathering these two pieces together, we derive the following dynamics for the corresponding cash process $\left\{ c_t\right\}_{0\leqslant t \leqslant T}:$  
\begin{equation} \label{eq:Cash}
    \begin{cases}
        dc_u =-p(u)\nu_u\,du - \kappa \nu_u^2\,du =-\left[ p(u) + \kappa \nu_u \right]\nu_u\,du,\\
        c_t = c.
    \end{cases}
\end{equation}
In \eqref{eq:Cash}, we denoted by $c \in \mathbb{R}$ the initial cash amount that the trader holds. We remark that we assume $\kappa$ to be the same among players. Since we suppose they are all trading the same asset in the same venue, this is not an unrealistic condition.

Finally, we define the agent's wealth $w$ as
\begin{equation} \label{eq:Wealth}
    w_u := c_u + q_u p(u) = c -\int_t^u \left[ p(\tau) + \kappa \nu_\tau\right] \nu_\tau \,d\tau + q_u p(u) \hspace{1.0cm} (t\leqslant u \leqslant T).
\end{equation}

\paragraph{\textbf{The value function}} We intend to approach the price formation in the MFG model via the dynamic programming perspective.  To conduct this endeavor, we define the value function $V$ as
\begin{align} \label{eq:ValueFn}
  \begin{split}
    V(t,q) &:= \sup_{\nu \in \mathcal{U}_t}\mathbb{E}\left[w_T - c - A\,q_T^2 - \phi \int_t^T q_u^2 \,du \Bigg| c_t = c,\, q_t = q \right]  \\
    &= \sup_{\nu \in \mathcal{U}_t}\mathbb{E}\left[q_T\left( p(T) - Aq_T \right) - \int_t^T \left\{\kappa \, \nu_u^2 +\phi\, q_u^2 +\nu_u p(u) \right\}\,du \Bigg| q_t = q \right].
  \end{split}
\end{align}
Here, $V$ is the supremum of performance criteria which are in turn formed by three parts. Firstly, they include the mean of the difference $w_T - c$ between the corresponding trader's terminal wealth and initial cash. We do not compare $w_T - c$ to the pre-trade price $q \, p(t)$ of the initial inventory,  a common enterprise in \textit{Implementation Shortfall} algorithms,  because we are not interested in taking differentials of the price. Secondly, the criteria also include a terminal penalization for finishing with non-vanishing inventory, i.e., $-A \, q_T^2.$ Thirdly, there is the mean of an urgency term $-\phi \int_t^T q_u^2\,du,$ inducing traders to accelerate their program. Analogously as authors in \cite{cartea2015algorithmic} explain, we emphasize that we do not interpret this as a financially meaningful penalization but rather as an extra risk management tool.  However, as put by \cite{duffie2017size} in a related modeling context, even if we assume traders are risk-neutral, they can incur costs for holding non-zero inventory, either long or short, which advocates in part of taking $\phi$ to be the same for all agents.

\paragraph{\textbf{Inventory density}} We denote the density of the law of the players' inventories at time $t$, $q_t$, by $m(t,\cdot)$. Additionally, we consider the initial density $m_0$ as given, i.e., $m(0,\cdot) = m_0$.

\paragraph{\textbf{Market clearing condition}} If an agent with inventory level $q$ trades at a rate $\nu(t,q),$ then we can define the aggregation rate $\mu$ as
$$
\mu(t) = \int_{\mathbb{R}} \nu(t,q) m(t,q)\,dq.
$$
Heuristically, being in equilibrium with a market demand/supply rate $\Lambda$ should require the market clearing condition, cf. \cite{gomes2020mean}:
\begin{equation} \label{eq:MktClearingMFG_nu}
    \mu(t) + \Lambda(t) = 0 \hspace{1.0cm} (0\leqslant t \leqslant T).
\end{equation}
Therefore, if the market is bullish (respectively, bearish) on aggregate, then we would expect $\mu(t) > 0$ (respectively, $\mu(t) < 0$), i.e., traders are required to buy (respectively, sell) some shares of the asset; hence, $\Lambda(t)$ describes the opposite movement, which occurs in the inventory of the traders' counterparts as a whole. We thus see $\Lambda$ as the interaction process of agents with the LOB itself (such as by adding or subtracting shares of an asset from there), or rather the resulting movement in the LOB due to the negotiations taking place there. We remark that we do not suppose that the players behaving optimally in our MFG trade among themselves --- in fact, players trading aggressively (meaning via market orders) could be, say, all buying the asset simultaneously at a certain time instant. Thus, we do implicitly regard that there are exogenous liquidity providers accepting to serve as counterparts at price $p$ to the mean-field of traders we are studying. 

\subsection{Competitive equilibria}

We are now ready to provide the following definition of competitive equilibria for the MFG model.

\begin{definition} \label{def:CompetitiveEquilibriaMFG}
A competitive equilibrium for the MFG model is a tuple $(p,V,m,\nu^*,\Lambda),$ where $p,\Lambda : \left[0,T\right] \rightarrow \mathbb{R}$ and $V,m,\nu^*:\left[0,T\right]\times \mathbb{R} \rightarrow \mathbb{R}$ satisfy:
\begin{itemize}
	\item[(a)] For each square-integrable random variable $\xi$, the SDE $dq^*_u = \nu^*(u,q^*_u)\,du + \sigma dW_u,\,q^*_t=\xi,$ has a strong solution;
	
	\item[(b)] For each $(t,q) \in [0,T] \times \mathbb{R}$, if we consider a solution of the SDE above with initial condition $q^*_t=q$ and write $\nu^* := \left\{ \nu^*_u := \nu^*(u,q^*_u) \right\}_{t \leqslant u \leqslant T}$, then the supremum in \eqref{eq:ValueFn} is attained by the strategy $\nu^*,$ and it is equal to $V(t,q);$
    \item[(c)] If the density of $\law(q^*_0)$ is $m(0,\cdot)$ and $\left\{ q^*_t \right\}_{0\leqslant t \leqslant T}$ is a solution of SDE in $(a)$ with such initial condition, then the function $m$ is the inventory distribution density of $q^*$, i.e. $m(t,\cdot)$ is the density of $\law(q^*_t)$, for $0 \leqslant t \leqslant T$; 
    \item[(d)] Equation \eqref{eq:MktClearingMFG_nu} holds for $\nu^*$ and $m$:
\begin{equation} \label{eq:MktClearingMFG}
    \int_{\mathbb{R}} \nu^*(t,q) m(t,q)\,dq + \Lambda(t) = 0 \hspace{1.0cm} (0\leqslant t \leqslant T).
\end{equation}
\end{itemize}
\end{definition}

On the one hand, in condition (b) of Definition \ref{def:CompetitiveEquilibriaMFG} is akin to the ``\textit{Profit Maximization}'' property of the equilibria \cite[Sections 2, 3, and 4]{kyle1985continuous}. On the other hand, we do not propose ``market efficiency'' in the context of our equilibria configurations. Rather, we leave a degree of freedom, viz., the order flow of liquidity taking orders. In our condition (d), we assume market makers clear the market, but following a perspective similar to that of \cite{kyle1985continuous},\footnote{``Modelling how market makers can earn the positive frictional profits necessary to attract them into the business of market making is an interesting topic which takes us away from our main objective of studying how price formation is influenced by the optimizing behavior of an insider in a somewhat idealized setting.'', see \cite[Section 2]{kyle1985continuous}.} i.e., we do not delve into the (interesting) question of their optimizing behavior in doing so. 

\subsection{Some properties of the competitive equilibria}

We now turn to the analysis of competitive equilibria in terms of Definition \ref{def:CompetitiveEquilibriaMFG}. The first step is to assume that the price component $p$ of the tuple $(p,V,m,\nu^*,\Lambda)$ is given and continuous and from this derive the well-posedness for $(V,m).$ Particularly, the linear-quadratic shape of the criteria in \eqref{eq:ValueFn} will allow us to show that $V$ is itself quadratic in the inventory variable. These comprise the content of the subsequent result.

\begin{proposition} \label{prop:HJBandFP}
Let us consider a competitive equilibrium $(p,V,m,\nu^*,\Lambda),$ with $p \in C\left( \left[0,T\right] \right).$ Then, $(V,m)$ must solve the system
\begin{equation} \label{eq:HJBandFP}
    \begin{cases}
        \partial_t V(t,q) +\frac{\sigma^2}{2}\partial_q^2 V(t,q) + \frac{1}{4\kappa}\left[ \partial_q V(t,q) - p(t) \right]^2 - \phi q^2 = 0,\\
        \partial_t m(t,q) - \frac{\sigma^2}{2}\partial_q^2 m(t,q) + \partial_q\left[ \frac{\left( \partial_q V(t,q) - p(t) \right)}{2\kappa} m(t,q) \right] =0, \\
        V(T,q) = q\left[ p(T) -Aq \right] \text{ and }  m(0,q) = m_0(q),
\end{cases}
\end{equation}
where we understand the first PDE above in the viscosity sense, and the second one in the distributional sense. Moreover, for each such $p,$ the system \eqref{eq:HJBandFP} admits a unique smooth solution $(V,m) \in \left[ C^{1,2} \right]^2,$ where
\begin{equation}\label{eq:ans_thet} 
V(t,q)=\theta_0(t)+\theta_1(t)q+\theta_2(t)q^2,
\end{equation} 
where
\begin{equation}\label{eq:sol_theta2}
\theta_2(t)=\sqrt{\kappa \phi}\frac{1-ce^{2\gamma(T-t)}}{1+ce^{2\gamma(T-t)}},\, \text{ for } c:=\frac{\sqrt{\kappa \phi} + A}{\sqrt{\kappa \phi} - A} \text{ and } \gamma:=\sqrt{\frac{\phi}{\kappa}},
\end{equation}
\begin{equation}\label{eq:sol_theta1}
\theta_1(t) = p(T)e^{\frac{1}{\kappa}\int_t^T \theta_2(s)ds}-\frac{1}{\kappa}\int_t^T \theta_2(s) p(s) e^{\frac{1}{\kappa}\int_t^s \theta_2(\tau)d\tau}ds,
\end{equation}
as well as
\begin{equation} \label{eq:sol_theta0}
    \theta_0(t) = \int_t^T \left[ \frac{1}{4\kappa}\left(\theta_1(u) - p(u) \right)^2 + 2\sigma^2\theta_2(u) \right]\,du.
\end{equation}
\end{proposition}
\begin{proof}
Given $p \in C\left(\left[0,T\right]\right),$ the value function $V,$ which we define according to \eqref{eq:ValueFn}, solves in the viscosity sense the Hamilton-Jacobi-Bellman equation comprising the first PDE in \eqref{eq:HJBandFP}, see \cite[Theorems 4.3.1 and 4.3.2, Remarks 4.3.4 and 4.3.5]{pham2009continuous} --- the terminal condition is clear from \eqref{eq:ValueFn}. Now, given $p$ and $V,$ we check (arguing as in  \cite[Lemma 3.3]{cardaliaguet2010notes}) that the density of players' inventory $m$ (weakly) solves the Kolmogorov-Fokker-Planck equation in \eqref{eq:HJBandFP}, with initial data $m(0,\cdot) = \law(q_0)$. To prove the second part of the Proposition, we make the ansatz
\begin{equation} \label{eq:Ansatz}
    V(t,q) = \theta_0(t) + \theta_1(t)q + \theta_2(t)q^2.
\end{equation}
Thus, the first equation in \eqref{eq:HJBandFP} becomes
\begin{equation}\label{eq:sys_thet}
\begin{cases}
\displaystyle \theta_2'(t)+\tfrac{1}{\kappa}\theta_2^2(t)-\phi=0,\\
\displaystyle \theta_1'(t)+\tfrac{1}{\kappa}\theta_2(t)(\theta_1(t)-p(t))=0,\\
\displaystyle \theta_0'(t)+\tfrac{1}{4\kappa}(\theta_1(t)-p(t))^2 + 2\sigma^2 \theta_2(t)=0,\\
\theta_0(T)=0, \theta_1(T)=p(T), \theta_2(T)=-A.
\end{cases}
\end{equation}

Solving the Ricatti equation determining $\theta_2$ is straightforward; it also appears in the standard Almgren-Chriss model, with only one trader, see \cite{cartea2015algorithmic}. With $\theta_2$ at hand, we derive the formula for $\theta_1$ by solving the second equation in \eqref{eq:sys_thet}, which is a linear ODE. Finally, we obtain $\theta_0$ from the third one via direct integration. At this point, we can carry out standard verification arguments to validate \eqref{eq:Ansatz}. With such $V \in C^{1,2}$ at hand, showing that the Fokker-Plack-Kolmogorov equation admits a unique smooth solution $m$ is standard, see \cite{cardaliaguet2010notes}.
\end{proof}

 
\begin{corollary} \label{cor:RepFeedbackControl}
If $(p,V,m,\nu^*,\Lambda)$ is a competitive equilibrium, then the optimal control in feedback form is given by 
\begin{equation}\label{eq:nu_theta}
\nu^*(t,q) = \frac{\partial_q V(t,q) - p(t)}{2\kappa} =  \frac{1}{2\kappa}\left(\theta_1(t) - p(t)\right)  +\frac{\theta_2(t)}{\kappa} q.
\end{equation}
\end{corollary} 
\begin{remark} \label{rem:NotMktClearing}
We notice that the market clearing condition \eqref{eq:MktClearingMFG} is not necessary to derive neither the characterization of $(V,m)$ we provided in Proposition \ref{prop:HJBandFP} nor relation \eqref{eq:nu_theta} we presented above.
\end{remark}
The identity \eqref{eq:nu_theta} is somewhat insightful, as it sheds light on how we can interpret the functions $\theta_1$ and $\theta_2$. In fact, on the one hand, since $\theta_2 < 0,$ a consequence of this equation is that $\theta_2 q$ is a contribution to the trading rate of a player holding inventory $q$ that represents a force driving her to clear these holdings. On the other hand, there is another term: $\theta_1 - p.$ If $\theta_1 > p$ (respectively, $\theta_1 < p)$ then this term leads the trader to buy (respectively, sell) shares. Therefore, we can interpret $\theta_1$ as some indication of the future value of the asset --- something that relation \eqref{eq:sol_theta1} itself suggests.

\subsection{Comparison with \cite{duffie2017size}}

Our goal is to understand how prices in competitive equilibria are determined by the interactions of traders taking place in the LOB. We will do so by expressing the price component $p$ of a competitive equilibrium $(p,V,m,\nu^*,\Lambda)$ in terms of $\Lambda.$ Prior to that, but already in this direction, let us fix such a tuple $(p,V,m,\nu^*,\Lambda).$ We observe that a full clearing for optimal traders is equivalent to require that $\Lambda \equiv 0.$ In this case, from the market clearing condition \eqref{eq:MktClearingMFG}, we would have 
$$
\int\limits_{\mathbb{R}} \nu^*(t,q) m(t,q)dq=0.
$$
In view of \eqref{eq:nu_theta}, this implies that 
\begin{equation} \label{eq:ComparingWDuffie}
    p(t)=\theta_1(t)+2\theta_2(t)E(t),
\end{equation}
where we have written $E(t) := \int_{q} q m(t,q)dq.$ In this way, we have recovered a formula which is much similar to those derived in \cite[Eq. (9)]{duffie2017size}, for the discrete double auction, and in \cite[Eq. (A.99)]{duffie2017size}, for the continuous case, even though they analyzed a framework considerably distinct \commY{from} ours. Particularly, we once again observe from Eq. \eqref{eq:ComparingWDuffie} the analogy between $\theta_1$ and a future expected value of the asset's price. However, here we do not have to necessarily constrain ourselves to $\Lambda \equiv 0.$ In fact, if we consider a given realized order flow $\Lambda,$ not necessarily vanishing, we obtain
$$
p(t)= 2 \kappa \Lambda(t)+\theta_1(t)+2 \theta_2(t) E(t).
$$ 
Nevertheless, at this moment, the influence of $\Lambda$ on $E$ is not completely clear. The main result of the present section will address this issue, and we will be able, to a certain extent, to parametrize competitive equilibria with respect to $\Lambda.$ Before advancing in this direction, we infer some representations for the population's behavior in an equilibrium configuration.

\begin{proposition}
Let us consider a competitive equilibrium $(p,V,m,\nu^*,\Lambda).$ In this setting, the optimal inventory holdings $\left\{ q^*_t \right\}_{0\leqslant t \leqslant T}$ of a trader beginning with $q_0 \in L^2(\Omega)$ shares is
\begin{equation}\label{eq:trajec_Qstar}
q^*_t=q_0 e^{\int_0^t\frac{\theta_2(t)}{\kappa}}+\frac{1}{2\kappa}\int_0^t (\theta_1(s)-p(s))e^{\int_s^t\frac{\theta_2(\tau)}{\kappa}d\tau}ds + \sigma \int_0^t e^{\int_s^t\frac{\theta_2(\tau)}{\kappa}d\tau} dW_s,
\end{equation}
for $t\in[0,T]$. Moreover, the optimal trading rate process is
\begin{equation} \label{eq:OptmStrat}
\nu^*_t=\frac{1}{2 \kappa}\left(\theta_1(t)-p(t)\right) + \frac{\theta_2(t)}{\kappa} q^*_t,
\end{equation}
where $\theta_1, \theta_2$ are given, respectively, in \eqref{eq:sol_theta2} and \eqref{eq:sol_theta1}.
\end{proposition}
\begin{proof}
The representation \eqref{eq:OptmStrat} of the optimal strategy follows promptly from \eqref{eq:nu_theta}. From \eqref{eq:Inventory} and \eqref{eq:nu_theta}, we derive
\[
dq^*_t =\frac{1}{2\kappa}(\theta_1(t)- p(t))\,dt + \frac{\theta_2(t)}{\kappa}q^*_t\,dt + \sigma dW_t.
\]
We easily see that the solution of this linear SDE is \eqref{eq:trajec_Qstar}. 
\end{proof}

\subsection{On the formed price in the MFG}\label{sec:ExplicitSolution}

We are now ready to state and prove the main theoretical result of the current section. Summing up, its content is a representation of the price in terms of the market supply/demand path in a competitive equilibrium, see \eqref{eq:FormedPriceMFG} below. Firstly, we provide the formula as a necessary condition for a given tuple to be a competitive equilibrium. We argue that this is the most relevant part of the theorem. In effect, this relation allows us to jump out of the model, looking for its validation on data stemming from real world prices of assets, which will be done in Section \ref{sec:Empirical}. 

\begin{theorem} \label{thm:MainTheorem}
Given a competitive equilibrium $(p,V,m,\nu^*,\Lambda)$ for the MFG, the formed price is
\begin{equation} \label{eq:FormedPriceMFG}
p(t) = p(0) + 2\phi E_0 t - 2\phi \int_0^t (t-u)\Lambda(u)du + 2 \kappa ( \Lambda(t) - \Lambda(0) ).
\end{equation}
\end{theorem}
\begin{proof}
Let us assume $(p,V,m,\nu^*,\Lambda)$ is a competitive equilibrium. Remember that the law of $q^*_t$ is given by $m(t,\cdot)$. We define:
\begin{equation}
  \begin{cases}
    \Pi(t):=\Ee[\partial_q V(t,q^*_t)] = \int_\Rr \partial_q V(t,q) m(t,q)\,dq; \\ 
    E(t):= \Ee [q^*_t] =\int_\Rr q m(t,q)\,dq.
  \end{cases}
\end{equation}

Differentiating the first equation in \eqref{eq:HJBandFP} with respect to $q$ yields
\begin{equation}\label{eq:diff_hjb}
\partial_q\partial_tV(t,q)+\frac{\sigma^2}{2}\partial_{q}^3V(t,q) +\frac{(\partial_qV(t,q) - p(t))}{2\kappa}\partial_{q}V(t,q)-2\phi q =0.
\end{equation} 
We emphasize that this differentiation is licit since $V$ is a quadratic polynomial with respect to $q,$ with coefficients that are differentiable in time. Therefore, we observe that
\begin{align}\label{eq:PiPrime1}
    \begin{split} 
        \Pi^\prime(t)&= \int_{\mathbb{R}} \partial_q\partial_t V(t,q)m(t,q)\,dq + \int_{\mathbb{R}} \partial_q V(t,q) \partial_t m(t,q)\,dq\\
        &=\int_{\mathbb{R}} \Big\{ -\frac{\sigma^2}{2}\partial_{q}^3V(t,q)-\frac{\left(\partial_q V(t,q) - p(t)\right)}{2\kappa}\partial_{q}^2V(t,q)+\phi q \Big\}m(t,q)\,dq\\
        &\hspace{3.5cm}-\int_{\mathbb{R}} \frac{1}{2\kappa}\partial_q V(t,q)\partial_q \left[ (\partial_q V(t,q)-p(t) )m(t,q) \right]\,dq.
    \end{split}
\end{align}
Next, we notice that carrying out integration by parts in identity \eqref{eq:PiPrime1}, we have
\begin{equation} \label{eq:PiPrime}
    \Pi^\prime(t)=2\phi E(t).
\end{equation}
Employing \eqref{eq:Ansatz} and \eqref{eq:nu_theta} in \eqref{eq:Inventory}, we infer
\begin{align*}
q^*_t &= q_0 + \int\limits_{0}^t \frac{\partial_q V(s,q^*_s) - p(s) }{2 \kappa}ds+  \sigma W_t,
\end{align*}
in such a way that
\begin{equation} \label{eq:RepresentationOfE}
    E(t)=\Ee[q_0] + \int\limits_{0}^t \Ee\bigg[ \frac{\partial_q V(s,q^*_s)  - p(s) }{2 \kappa}\bigg]ds.
\end{equation}
From \eqref{eq:RepresentationOfE} and \eqref{eq:MktClearingMFG}, we deduce
\begin{equation} \label{eq:Eprime}
    E^\prime(t)=\Ee\bigg[ \frac{\partial_q V(t,q^*_t)  - p(t) }{2 \kappa}\bigg] = -\Lambda(t),
\end{equation}
and additionally, by the definition of $\Pi$, we find
\begin{equation} \label{eq:Eprime_Pi}
    E^\prime(t)=\Ee\bigg[ \frac{\partial_q V(t,q^*_t)  - p(t) }{2 \kappa}\bigg] = \frac{\Pi(t) - p(t)}{2\kappa}.
\end{equation}
Collecting \eqref{eq:PiPrime}, \eqref{eq:Eprime} and \eqref{eq:Eprime_Pi}, we obtain
\begin{equation}\label{eq:system_p}
\begin{cases}
\Pi^\prime(t)&=2\phi E(t),\\[10pt]
E^\prime(t)&= \displaystyle \frac{\Pi(t)}{2\kappa}-\frac{p(t)}{2\kappa}= -\Lambda(t).\\[10pt]
\end{cases}
\end{equation}
On the one hand, from the second equation in system \eqref{eq:system_p}, we have
\begin{equation*} 
    E(t) = E_0 - \int_0^t \Lambda(s)ds.
\end{equation*}
On the other hand, the first one gives us
\begin{align*} 
    \begin{split}
        \Pi(t) &= \Pi(0) + 2\phi \int_0^t  E(s)ds\\
        &= \Pi(0) + 2\phi E_0 t - 2\phi \int_0^t \int_0^s\Lambda(u)\,du\,ds\\
        &=\Pi(0) + 2\phi E_0 t - 2\phi \int_0^t \int_u^t\Lambda(u)\,ds\,du\\
        &=\Pi(0) + 2\phi E_0 t - 2\phi \int_0^t (t-u)\Lambda(u)\,du,
    \end{split}
\end{align*}
as well as, again from the second equation,
\begin{align} \label{eq:PandPi}
    \begin{split}
        p(t) &= \Pi(t) + 2 \kappa \Lambda(t)\\
        &= \Pi(0) + 2\phi E_0 t - 2\phi \int_0^t (t-u)\Lambda(u)du + 2 \kappa \Lambda(t).
    \end{split}
\end{align}
In particular,
\begin{equation} \label{eq:Pi0}
    \Pi(0) = p(0)-2 \kappa \Lambda(0). 
\end{equation}
In conclusion, \eqref{eq:PandPi} and \eqref{eq:Pi0} allow us to infer
\begin{equation}
p(t) = p(0) + 2\phi E_0 t - 2\phi \int_0^t (t-u)\Lambda(u)du + 2 \kappa ( \Lambda(t) - \Lambda(0)).
\end{equation}

\end{proof}

\begin{remark}\label{rmk:SDE_for_p_quadratic_var}

The expression for the formed price \eqref{eq:FormedPriceMFG} implies the dynamics
\begin{equation} \label{eq:diff_eqn_p_lambda_MFG}
dp(t) = 2\phi \left(E_0 - \int_0^t \Lambda(u)du\right) dt + 2\kappa d\Lambda(t).
\end{equation}
Notice the drift of $p$ depends on the path of $\Lambda$ through its integral. Moreover, one might use this dynamics to conclude the following correspondence of quadratic variations: $\langle p \rangle_t = 4\kappa^2 \langle \Lambda \rangle_t.$

\end{remark}

\subsection{Discussion on the formed price}


Let us discuss the main aspects of the formula \eqref{eq:FormedPriceMFG}. 

We first consider $2\kappa(\Lambda(t) - \Lambda(0))$ --- it corresponds to the contribution of the current instantaneous market supply/demand $\Lambda(t)$ rate to the formed price (at time $t$). We remark that it is proportional to the difference $\Lambda(t) - \Lambda(0),$ the slope being $2\kappa.$ This is the component of the price which results from the fact that traders are willing to pay costs for trading the asset. In effect, for trading a volume of $\left|\Lambda(t) - \Lambda(0)\right|,$ they will pay a total of $\kappa\left| \Lambda(t) - \Lambda(0)\right|^2$ in transaction costs due to this market friction. The resulting influence of this movement on the formed price is $2\kappa\left(\Lambda(t) - \Lambda(0) \right)$.

A similar consideration is valid for $2\phi E_0 t.$ Namely, if the average initial holdings (or, equivalently, the initial targets) of players, $E_0$, is positive (resp., negative), then our modeling implies they are willing to sell (resp., buy). However, this term $2\phi E_0 t$ indicates a positive (resp., negative) drift on the price, as long as $\phi > 0,$ which grows proportionally to time. Thus, if players are not quick enough to execute their schedules, we would expect prices to be moving favorably to their current holdings, with this effect being more pronounced for larger times. In a sense, this term represents an inertia component on the asset's price.  Here, the inertia we refer to is distinct from what was studied in \cite{bayraktar2007queuing} (see also references therein).  In our setting, this terminology only indicates that $2 \phi E_0 t$ exert a sort of resistance on trading rates moving the price.

Finally, we comment on $-2\phi \int_0^t (t-u) \Lambda(u) du$, which is undoubtedly the most interesting term. We notice we can naturally interpret it as a \textit{memory} term. Indeed, the weight $(t-u)$ of the value $\Lambda(u)$ appearing within the integral is larger for older market supply/demand rate values.  In a certain sense, as markets ``digest'' past information, they will appear more significantly in the formed price formula.

We finish the present section by carrying out some explicit comparative statics to highlight the influence of the parameters on the formed price \eqref{eq:FormedPriceMFG}.

\paragraph{\textbf{Sensitivity relative to the temporary price impact}} We observe that $p$ is linear in $\kappa,$ with slope $\partial_{\kappa} p(t) = 2(\Lambda(t) - \Lambda(0)).$ As we discussed above, this is the term in the price representing the trader's willingness to pay costs to negotiate the asset. 

\paragraph{\textbf{Sensitivity relative to the average initial target}} We notice that $p$ is also linear with respect to $E_0.$ In this case, we have slope $\partial_{E_0} p(t) = 2\phi t,$ which always has the same sign as $\phi.$ Thus, if traders are risk averse, meaning $\phi > 0$ (resp., risk loving, meaning $\phi<0$), then $p$ is an increasing (resp., decreasing) function of $E_0.$ Let us assume for concreteness that $\phi > 0.$ Then, if the market has a larger target, \textit{ceteris paribus} --- particularly, while $\Lambda$ remains still --- we can only explain this by the price being larger as well. In other words, if we consider that, when behaving according to the same criteria, players still trade with all else remaining the same, especially demanding on average the same $\Lambda$, then the formed price must be larger. 

\paragraph{\textbf{Sensitivity relative to the trading urgency}} The price $p$ is also linear with respect to $\phi.$ The slope here is $\partial_{\phi} p(t) = 2E_0 t - 2\int_0^t (t-u)\Lambda(u)\,du.$ Thus, the conflict between the memory term, representing the detrimental pressure traders undergo due to market frictions, and one of the inertia terms, comprising the average initial holdings/targets of players, governs how prices change proportionally to the urgency parameter. 

\paragraph{\textbf{Sensitivity relative to the order flow}} Finally, we analyze how prices depend on $\Lambda.$ They are linear on it, in a functional sense. Namely, given a time $t\in\left[0,\,T\right],$ as a function of the path $\left\{ \Lambda(u) \right\}_{0\leqslant u \leqslant t},$ the formed price $p(t)$ evaluated at time $t$ is affine. We identify the analogous element to the slope here (which is now a whole path) if we perturb the order flow $\Lambda$ by trading at a rate $\left\{ \nu_u \right\}_{0 \leq u \leq t},$ and this yields the difference between the two resulting prices:
$$
2\kappa \nu_t - 2\phi \int_0^t (t-u)\nu_u \,du.
$$

\subsection{Existence of Equilibrium}


In this section we discuss two settings. Basically, among the realized order flow $\Lambda$ and the formed price $p,$ there is only one degree of freedom. Namely, assuming the price $p$ is given, with suitable regularity, we can derive the order flow in terms of the traders' optimal turnover rate $\nu^*$ and their inventory distribution density $m$ --- this is the content of our next result. Conversely, if we assume that the order flow is given, and that it satisfies suitable compatibility conditions (which we will describe later, see \eqref{eq:Compatibility}), then we can assert that $p,$ determined by the formula \eqref{eq:FormedPriceMFG}, is \textit{the} price formed by this order flow. In the latter framework, the necessity clause stems precisely from \eqref{eq:FormedPriceMFG}, whereas sufficiency is what demands the compatibility relation we mentioned above. We proceed to discuss these results in the sequel.

Before doing so, let us exemplify the case when $p$ is given considering  that $p(t) \equiv p_0 $ is a constant, with $p_0 > 0.$ Then, we can derive from \eqref{eq:diff_eqn_p_lambda_MFG} that $\Lambda$ solves 
$$
\begin{cases}
\kappa \Lambda^{\prime\prime}(t) = \phi \Lambda(t),\\
\Lambda(0) = \Lambda_0,\, \Lambda^\prime(0) = \phi E_0 / \kappa.
\end{cases}
$$
If we assume that the model parameters satisfy
$$
\Lambda_0\sqrt{\kappa} - \sqrt{\phi} E_0 \neq 0,
$$
then we can derive from the above second-order ODE the closed-form solution for $\Lambda:$
$$
\Lambda(t) = \Lambda_0 e^{\int_0^t g(u)\,du},
$$
where
$$
g(u) = \gamma\left( \frac{e^{2\gamma u}\beta - 1}{e^{2\gamma u}\beta + 1} \right), u \geqslant 0,
$$
the constants $\gamma$ and $\beta$ being given by
$$
\gamma = \sqrt{\frac{\phi}{\kappa}} \text{ and } \beta= \frac{\Lambda_0\sqrt{\kappa} + \sqrt{\phi} E_0}{\Lambda_0\sqrt{\kappa} - \sqrt{\phi} E_0}.
$$


The more general corollary below follows from Proposition \ref{prop:HJBandFP}.

\begin{corollary}
Let $p \in C\left( \left[0,T\right] \right),$ $(V,m)$ be solution to \eqref{eq:HJBandFP}, and $\nu^*$ be given in feedback form by \eqref{eq:nu_theta}. Then, by setting $\Lambda(t) = -\int_q \nu^*(t,q)m(t,q)dq,$ the tuple $(p,V,m,\nu^*,\Lambda)$ is a competitive equilibrium. In particular, competitive equilibria exist.
\end{corollary}

We now consider the other situation where $\Lambda$ is given and derive a condition for the price given by \eqref{eq:FormedPriceMFG} that gives us a competitive equilibrium.

\begin{proposition}
Given $\Lambda,$ let us define $p$ as in \eqref{eq:FormedPriceMFG}, $(V,m)$ as the corresponding solution of \eqref{eq:HJBandFP}, and $\nu^*$ as in \eqref{eq:nu_theta}. Then, the compatibility condition\footnote{Notice that $\theta_i$ depends on $p$ for $i=1,2$.}
\begin{equation} \label{eq:Compatibility}
    \theta_1(t) + 2\theta_2(t) \int_{\mathbb{R}} qm(t,q)dq = p(0) + 2\phi E_0 t - 2\phi \int_0^t (t-u)\Lambda(u)du - 2\kappa \Lambda(0)
\end{equation}
implies that $(p,V,m,\nu^*,\Lambda)$ is a competitive equilibrium.
\end{proposition}

\begin{proof}
Let us assume $\Lambda$ is given, that we define $p$ as \eqref{eq:FormedPriceMFG}, that $(V,m)$ is the solution of \eqref{eq:HJBandFP} corresponding to $p,$ and that $\nu^*$ satisfies \eqref{eq:nu_theta}. At this point, to show that $(p,V,m,\nu^*,\Lambda)$ is a competitive equilibrium, we need only to verify the market clearing condition \eqref{eq:MktClearingMFG}. To do this, we integrate \eqref{eq:nu_theta} against $m$ (cf. Remark \ref{rem:NotMktClearing}), whence we derive
\begin{align*}
    \int_{\mathbb{R}} \nu^*(t,q)&m(t,q)dq = \frac{1}{2\kappa}\left( \theta_1(t) - p(t) \right) + \frac{\theta_2(t)}{\kappa}\int_{\mathbb{R}} q m(t,q)dq \\
    &= \frac{\theta_1(t)}{2\kappa} + \frac{\theta_2(t)}{\kappa}\int_{\mathbb{R}} q m(t,q)dq - \frac{1}{2\kappa}\left[ p(0) + 2\phi E_0 t - 2\phi \int_0^t (t-u)\Lambda(u)du + 2 \kappa ( \Lambda(t) - \Lambda(0)) \right] \\
    &= \frac{\theta_1(t)}{2\kappa} + \frac{\theta_2(t)}{\kappa}\int_{\mathbb{R}} q m(t,q)dq - \frac{1}{2\kappa}\left[ p(0) + 2\phi E_0 t - 2\phi \int_0^t (t-u)\Lambda(u)du - 2 \kappa \Lambda(0) \right] - \Lambda(t). \\
    &= - \Lambda(t),
\end{align*}
where we used \eqref{eq:Compatibility} in the last equality. 
\end{proof}

The ``compatibility condition'' \eqref{eq:Compatibility} ensures that, if we \textit{define} the price according to the necessary relation, then the corresponding tuple (consisting of the appropriate objects, viz., price, value function, density of inventories, strategy in feedback form, and market supply/demand path) is indeed a competitive equilibrium.

\section{Price formation with a finite population model} \label{sec:FinitePop} 

In the price discovery framework, a relevant research question is to investigate the convergence of the short-term price to the ``true'' long-term one. In effect, as \cite{biais2005market} points out, in reference to \cite{walras1896elements}, the tradition L. Walras initiated has roots on this question.\footnote{``In perfect markets, Walrasian equilibrium prices reflect the competitive demand curves of all potential investors.  While determining these fundamental equilibrium valuations is the focus of (most) asset pricing, market microstructure studies how, in the short term, transaction prices converge to (or deviate from) long-term equilibrium values. Walras himself was concerned about the convergence to equilibrium prices, through a t\^atonnement process. One of the first descriptions of the microstructure of a financial market can be found in the Elements d’Economie Politique Pure (1874), where he describes the workings of the Paris Bourse. Walras’s field observations contributed to the genesis of his formalization of how supply and demand are expressed and markets clear. Market microstructure offers a unique opportunity to confront directly microeconomic theory with the actual workings of markets. This facilitates both tests of economic models and the development of policy prescriptions.''} We focus on a transposed version of this question in the price formation framework in the current section.  More precisely, we inquire under which conditions we can ensure that the price which is formed as a result of the interaction among traders of a finite population converges to the one that we have identified in the mean-field setting, i.e., the one we found in \eqref{eq:FormedPriceMFG}.


\subsection{The model} \label{subsec:MktModel}

\paragraph{\textbf{The market}} We consider a market consisting of a single venue. There are $N\geqslant 1$ traders present there, trading a single asset, in similar terms as in Section \ref{sec:MFG}. We index players using labels $i \in \mathcal{N} := \left\{ 1,\ldots, N \right\}.$ We denote the mid-price of the asset by $p = \left\{ p_t \right\}_{0\leqslant t \leqslant T}.$ Here, we allow $p$ to be stochastic. We will discuss our assumptions on $p$ below. \\

Before doing this, we describe the information structure we assume throughout the present section.

\paragraph{\textbf{Information structure}} We fix a complete probability space $(\Omega, \Ff, \mathbb{P})$, and we consider a $N+1$-dimensional Brownian motion $W = \left\{ (W^0_t, W^1_t, \ldots, W^N_t) \right\}_{0\leqslant t \leqslant T}$. Additionally, we denote the filtration generated by $W^0$ and $W^i$ by $\mathbb{F}^{0i} = \left\{ \mathcal{F}^{0i}_t \right\}_{0\leqslant t \leqslant T}$, which satisfies the usual conditions and models the gathering of information of player $i.$ We write $\mathbb{F}^0 := \left\{ \mathcal{F}^0_t \right\}_{0 \leqslant t \leqslant T}$ for the common information shared by all the players, which in our case is the filtration generated by $W^0$.

\paragraph{\textbf{Admissible prices}} In the finite population setting, we look for prices that are sensitive to the movements of each particular individual. That is, we consider possible prices $p$ that are mappings depending on the strategies of each trader. To carry out this endeavor, we will first specify the class of admissible strategies. Henceforth, for each $i \in \mathcal{N} \cup \left\{ 0 \right\},$ the set $\mathcal{A}_i$ consists of the $\mathbb{F}^{0i}-$progressively measurable processes $\left\{ \nu_t\right\}_{0\leqslant t \leqslant T}$ such that $\mathbb{E}\left[\int_0^T  \nu_t^2\,dt \right] < \infty.$ Also, we put 
$$
\mathcal{A} := \Pi_{i=1}^N \mathcal{A}_i, \text{ and } \mathcal{A}_{-i} := \Pi_{j\neq i} \mathcal{A}_j,
$$
For each $\boldsymbol{\nu} = \left( \nu^1,\ldots, \nu^N \right)^\intercal \in \mathcal{A}$ and $i\in\mathcal{N},$ we will write $\boldsymbol{\nu}^{-i} := \left( \nu^1,\ldots,\nu^{i-1},\nu^{i+1},\ldots,\nu^N\right)^\intercal,$ as well as $\left( \nu^i; \boldsymbol{\nu}^{-i} \right) := \boldsymbol{\nu}.$

Then, we focus on the class of prices $\mathfrak{P}$ comprising the mappings $\boldsymbol{\nu} \in \mathcal{A} \mapsto p\left( \boldsymbol{\nu} \right) \in \mathcal{A}_0$ satisfying the subsequent conditions:
\begin{itemize}
    \item[(\textbf{C1})] For each $t \in \left[0,T\right],$ $i \in \mathcal{N}$ and $\bm{\nu} \in \mathcal{A},$ there exists a process $\left\{\xi_{t,u}^{i,N}(\bm{\nu})\right\}_{0\leqslant u \leqslant t} \in \mathcal{A}_i$ such that
    $$
    \left\langle D_ip_t(\bm{\nu}),w^i\right\rangle := \lim_{\epsilon \rightarrow 0} \frac{p_t(\nu^i + \epsilon w^i;\boldsymbol{\nu}^{-i}) - p_t(\nu^i; \boldsymbol{\nu}^{-i})}{\epsilon} = \mathbb{E}\left[ \int_0^t \xi_{t,u}^{i,N}(\bm{\nu}) w^i_u\,du \Big| \mathcal{F}^0_t \right] \hspace{1.0cm} (w^i \in \mathcal{A}_i),
    $$
    the above limit holding strongly in the $L^2(\Omega)-$topology.
    \item[(\textbf{C2})] For each $t \in \left[0,T\right],$ $i \in \mathcal{N}$ and $\bm{\nu} \in \mathcal{A},$ we have $ \xi^{i,N}_{T,t}(\bm{\nu}) \in L^2(\Omega, \mathcal{F}^0_T).$ 
\end{itemize}

Due to the structure of the influence of an individual trading rate on the asset price, we call $\xi_{t,u}^{i,N}$ the microstructural kernel impact of player $i$ on the asset price.

Next, we provide an example of prices that belong to the class $\mathfrak{P},$ whence it is non-empty.
\begin{example} \label{ex:PriceExample}
If $p$ is given by 
\begin{equation} \label{eq:examplePrice}
    p_t(\bm{\nu}) = \widetilde{p}_t + \alpha\mathbb{E}\left[ \frac{1}{N}\sum_{i=1}^N q^i_t \Bigg| \mathcal{F}^0_t \right],
\end{equation}
for some $\alpha > 0,$ and some process $\{\widetilde{p}_t\}_{0\leqslant t \leqslant T} \in \mathcal{A}_0,$ then $p \in \mathfrak{P}.$ In this case, we have $\xi^{i,N}_{t,u}(\bm{\nu}) \equiv \alpha/N.$
\end{example}
Note that Example \ref{eq:examplePrice} corresponds to the case in which the reference price is linearly affected by the traders' turnover rates -- it is a generalized formulation of a linear permanent price impact.

We elaborate further discussions on $\mathfrak{P}$ in \ref{app:discAdmPrices}.


\paragraph{\textbf{The traders}} Let us fix $p \in \mathfrak{P}.$ For each $i \in \mathcal{N},$ the trader $i$ controls her turnover rate $\left\{ \nu^i_t \right\}_{0\leqslant t \leqslant T},$ selecting it out of the set $\mathcal{A}_i.$ Consequently, her inventory process, which we denote by $\left\{ q_t^i \right\}_{0\leqslant t \leqslant T},$ has the dynamics 
\begin{equation}\label{eq:inv_opt_trad}
\begin{cases}
\displaystyle dq^i_t = \nu^i_t dt + \sigma^i d W^i_t,\\
q^i_0 \text{ given,}
\end{cases}
\end{equation}
where $\sigma^i\geq 0$ for each $i\in\mathcal N.$ With a similar construction as in \eqref{eq:Cash}, we stipulate that her cash process $\left\{ c^i_t \right\}_{0 \leqslant t \leqslant T}$ evolves according to
\begin{equation}\label{eq:cash}
\begin{cases}
\displaystyle dc^i_t = -\left[ p_t\left( \boldsymbol{\nu} \right) + \kappa \nu_t^i \right]\nu_t^i\,dt ,\\
c^i_0 \text{ given.}
\end{cases}
\end{equation}
Furthermore, by denoting the wealth process of this trader by $\left\{ w_t^i := c^i_t + q^i_t p_t(\boldsymbol{\nu}) \right\}_{0\leqslant t \leqslant T},$ we have, as in \eqref{eq:Wealth},
\begin{align} \label{eq:wealth}
    \begin{split}
        w_t^i &=c^i_0 -\int_0^t \left[p_u\left(\boldsymbol{\nu}\right) + \kappa \nu^i_u\right] \nu^i_u du + q^i_t p_t\left(\boldsymbol{\nu}\right),
    \end{split}
\end{align} 
for each $t \in \left[0,T\right].$

\paragraph{\textbf{The performance criteria}} Following a similar reasoning as in the construction of \eqref{eq:ValueFn}, we build the performance criteria for player $i.$ Assuming agent $i$ uses strategy $\nu^i \in \mathcal{A}_i,$ and that the remaining ones trade with a profile $\boldsymbol{\nu}^{-i} \in \mathcal{A}_{-i},$ we define the payoff $J_i\left( \nu^i; \boldsymbol{\nu}^{-i}\right)$ of player $i$ as
\begin{align}\label{eq:cost1}
  \begin{split}
    J^i(\nu^i;\bm{\nu}^{-i}) &:= \Ee \bigg[ w_T^i - c^i_0 -A \, (q_T^i)^2-\phi\int_0^T (q_u^i)^2 du \bigg] \\ 
    &= \Ee \bigg[-\int_0^T \bigg[\kappa \, (\nu_u^i)^2 + \phi \, (q_u^i)^2 +\nu_u^i p_u\left(\boldsymbol{\nu}\right)  \bigg] du + q_T^i ( p_T - A q_T^i )\bigg].
  \end{split}
\end{align}

\subsection{Competitive equilibria}

In this section, we seek a definition of competitive equilibria for the $N-$person game analogous to the one we made for the MFG, i.e., Definition \ref{def:CompetitiveEquilibriaMFG}. Regarding the optimizing behavior of the traders in the population, it is natural to expect that they will accommodate to a Nash equilibrium, since we want them to be competing and rational. A precise definition, in the present context, is the content of the subsequent definition.
\begin{definition}\label{def:nasheq}
	We say that an admissible strategy profile ${\bm{\nu}}^*=\left(\nu^{*1}, \cdots, \nu^{*N}\right)^\intercal \in \mathcal{A}$ is a Nash equilibrium for the $N-$person game if, for each $i \in \mathcal{N}$ and $\nu^{i} \in \mathcal{A}_i,$ 
	\begin{equation}\label{eq:nasheq}
	\quad J^{i}({\bm{\nu}^*}) \geqslant J^{i}\left(\nu^{i}; {\bm{\nu}^*}^{-i}\right)
	\end{equation}
\end{definition}

The next element we seek is a counterpart of the market clearing condition \eqref{eq:MktClearingMFG} to the present setting. Thus, we consider a stochastic process $\Lambda = \left\{ \Lambda_{t} \right\}_{0\leqslant t \leqslant T}$ representing the market supply rate. In this case, if player $i \in \mathcal{N}$ considers strategy $\nu^i \in \mathcal{A}_i,$ the following market clearing condition must be valid:
\begin{equation} \label{eq:MktClearing1}
    \mathbb{E}\left[ \frac{1}{N}\sum_{i=1}^N \nu^i_t \Bigg| \mathcal{F}^0_t \right] + \Lambda_{t} = 0 \hspace{1.0cm} (0 \leqslant t \leqslant T).
\end{equation}

We are ready to propose a definition of a competitive equilibrium for the $N-$person game.
\begin{definition}
A triplet $\left(p,{\bm{\nu}}^*,\Lambda\right) \in \mathfrak{P} \times \mathcal{A} \times \mathcal{A}_0 $ is a competitive equilibrium under the performance criteria $\left\{ J^i \right\}_{i=1}^N$ if:
\begin{itemize}
    \item[(i)] Given $p,$ the process ${\bm{\nu}}^*$ is a Nash equilibrium under the criteria $\left\{J^i\right\}_{i=1}^N;$
    \item[(ii)] The processes ${\bm{\nu}}^*$ and $\Lambda$ together solve \eqref{eq:MktClearing1}.
\end{itemize}
\end{definition}

\subsection{Some properties of the competitive equilibria} \label{subsec:PriceFormationFinitePop}

Using variational techniques,  as in \cite{casgrain2020mean} and \cite{evangelista2020finite}, we can derive the subsequent necessary conditions for a Nash equilibrium. For its proof, see \ref{app:proof}.


\begin{proposition} \label{prop:CharacterizationOfTheNE}
Assume that, for each $i \in \mathcal{N}$ and $\bm{\nu}^{-i},$ the functional $\nu^i \mapsto J^i(\nu^i;\bm{\nu}^{-i})$ admits a unique maximizer. Then, a Nash equilibrium $\bm{\nu}^*$ for the criteria $\left\{ J^1,...,J^N \right\}$ must solve the FBSDE
\begin{equation} \label{eq:FBSDE}
    \begin{cases}
        q^{*i}_t = q^i + \int_0^t \nu^{*i}_u\,du + \sigma^i W^i_t,\\
        2\kappa\,\nu^{*i}_t = -p_t(\nu^{*i};\bm{\nu}^{*-i}) + 2\phi\,\int_0^t q^{*i}_u\,du \\
        \hspace{1.2cm} + \mathbb{E}\left[ \xi^{i,N}_{T,t}(\nu^{*i};\bm{\nu}^{*-i})q^{*i}_T - \int_t^T \xi^{i,N}_{u,t}(\nu^{*i};\bm{\nu}^{*-i})\nu^{*i}_u\,du \Big| \mathcal{F}^{0i}_t \right] + M^{*i}_t,
    \end{cases}
\end{equation}
for suitable square-integrable martingales $\{ M^{*i}_t,\mathcal{F}^{0i}_t \}_{0\leqslant t \leqslant T}$ which satisfy 
$$
M^{*i}_T = -2A q^{*i}_T + p_T(\nu^{*i};\bm{\nu}^{*-i}) - 2\phi\int_0^T q^{*i}_t\,dt,
$$
$\mathbb{P}-$almost surely.
\end{proposition}

\begin{remark}
The common information shared by all the players appears, for instance, into the martingale $M^{*i}$. Under the Brownian assumption, one could use the Martingale Representation Theorem to write this term as a combination of stochastic integrals with respect to $W^i$ and $W^0$.
\end{remark}

It is illustrative to consider Proposition \ref{prop:CharacterizationOfTheNE} in the case of permanent price impact. We comment on it in the subsequent example.
\begin{example} \label{eq:ExampleLinearImpact}
In the linear permanent price impact case, i.e., $p_t\left(\bm{\nu}\right) = \widetilde{p}_t + \frac{\alpha}{N}\sum_{i=1}^N q^i_t,$ with $\alpha >0,$ $\widetilde{p} \in \mathcal{U}_0$ being a semimartingale, and $\widetilde{p}_0,\,\widetilde{p}_T \in L^2(\Omega),$ we have
$$
J^i(\bm{\nu}) = \mathbb{E}\left[ \int_0^T \left\{ \left[-\kappa \left( \nu^i_t \right)^2 - \phi\left(q^i_t\right)^2 + \frac{b}{N}q^i_t\sum_{j\neq i}\nu^j_t \right]\,dt + q^i_t\,d\widetilde{p}_t \right\} - \left( A - \frac{\alpha}{2N} \right)\left(q^i_T\right)^2 + q^i\widetilde{p}_0 \right] .
$$
Thus, assuming $\kappa>0$ and $A>\alpha/(2N),$ we conclude that the functional $\nu^i \in \mathcal{A}_i \mapsto J_i(\nu^i;\bm{\nu}^{-i}) \in \mathbb{R}$ is strictly concave and coercive, for each $i \in \mathcal{N}$ and $\bm{\nu}^{-i} \in \mathcal{A}_{-i};$ hence, it admits a unique maximizer. Conversely, employing the techniques we used in \cite{evangelista2020finite} (with suitable \commY{integrability} assumptions on $\widetilde{p}$), we can prove that \eqref{eq:FBSDE} admits a unique strong solution.
\end{example}

\subsection{On the formed price in the finite population game}

In Proposition \ref{prop:CharacterizationOfTheNE}, we have given necessary conditions for the Nash equilibrium resulting from a price process $p \in \mathfrak{P}.$ In Example \ref{eq:ExampleLinearImpact} above, we provided a class of prices for which we can, in fact, prove sufficient conditions, thus characterizing Nash equilibria in this circumstance. The next result we obtain is the main one in this section. We use this necessary conditions to establish a formula analogous to \eqref{eq:FormedPriceMFG} in the current finite population framework in a competitive equilibrium. Some additional elements will appear due to the microstructural impact that we allow in prices,  and the consequence of asymmetry in the information that results from the finite population of traders. We turn to the investigation of this question in what follows.
%

\begin{theorem} \label{thm:PriceFormed}
Let us consider a competitive equilibrium $\left(p^N,\boldsymbol{\nu}^*,\Lambda^N\right)$ for the $N-$person game. We assume that there exists a mapping $\Lambda \mapsto \left\{ \xi_{t,u}^N(\Lambda) \right\}_{0\leqslant u \leqslant t \leqslant T}$\footnote{Notice the slight abuse of notation since this process $\xi_{t,u}^N(\Lambda)$, which does not depend on $i$, is not the process $\left\{\xi_{t,u}^{i,N}(\bm{\nu})\right\}_{0\leqslant u \leqslant t}$ from condition (C1) above.} such that, for every $i \in \mathcal{N}$ and $\boldsymbol{\nu} \in\mathcal{A},$ 
\begin{equation} \label{eq:Assumption2}
    \xi_{t,u}^{i,N}(\boldsymbol{\nu}) = \xi^N_{t,u}\left( -\mathbb{E}\left[ \frac{1}{N}\sum_{i=1}^N \nu^i_{\cdot} \Bigg| \mathcal{F}^0_\cdot \right] \right) \hspace{1.0cm} (0\leqslant u \leqslant t \leqslant T).
\end{equation} 
Then, the price $\{p_t^N = p_t(\bm{\nu}^*) \}_{0\leqslant t \leqslant T}$ is formed as
\begin{align} \label{eq:FormedPrice}
    \begin{split}
        p_t^N =&\, p^N_0 + 2\kappa\left(\Lambda^N_{t} - \Lambda^N_{0}\right) +2\phi \int_0^t E^N_{0,u}\,du - 2\phi \int_0^t (t-u) \Lambda^N_{u}\,du + \int_0^t \epsilon^N_u\,du + F^N_{t} + M^N_{t},
    \end{split}
\end{align}
where the process $\left\{F^N_{t} \right\}_{0\leqslant t \leqslant T}$ is given by
\begin{equation} \label{eq:cStarN}
  \begin{cases}
    F^N_{t} := f^N_{t} - f^N_{0}, \\
    f^N_{t} := \mathbb{E}\left[ \xi_{T,t}^N\left( \Lambda^N \right)\left(E^N_{0,T} - \int_0^T \Lambda^N_{u}\,du + \epsilon^N_T \right) + \int_t^T \xi_{u,t}^N\left(\Lambda^N\right)\Lambda^N_{u}\,du \Bigg| \mathcal{F}^0_t \right] \hspace{1.0cm} (0\leqslant t \leqslant T),
  \end{cases}
\end{equation}
whereas $\left\{\epsilon^N_t, \mathcal{F}^0_t \right\}_{0\leqslant t \leqslant T},$ $\left\{ E^N_{0,t},\mathcal{F}^0_t \right\}_{0\leqslant t \leqslant T}$ and $\left\{ M^N_{t}, \mathcal{F}^0_t \right\}_{0\leqslant t \leqslant T}$ are square-integrable martingales given by
\begin{equation} \label{eq:Martingale1}
    \epsilon^N_t :=  \mathbb{E}\left[ \frac{1}{N}\sum_{i=1}^N \sigma^i W^i_t\Bigg| \mathcal{F}^0_t\right] + \frac{1}{N}\sum_{i=1}^N \int_0^t \left\{\mathbb{E}\left[ \nu^{*i}_s | \mathcal{F}^0_t \right] - \mathbb{E}\left[ \nu^{*i}_s | \mathcal{F}^0_s \right] \right\}\,ds ,
\end{equation}
\begin{equation} \label{eq:Martingale2}
    E^N_{0,t} := \mathbb{E}\left[ \frac{1}{N}\sum_{i=1}^N q^i_0 \Bigg| \mathcal{F}^0_t \right],
\end{equation}
and
\begin{equation} \label{eq:Martingale3}
    M^N_{t} := \mathbb{E}\left[ \frac{1}{N}\sum_{i=1}^N \left( M^{*i}_t - M^{*i}_0 \right) \Bigg| \mathcal{F}^0_t \right] + \int_0^t \left\{ \mathbb{E}\left[ \frac{1}{N}\sum_{i=1}^N q^{*i}_s \Bigg| \mathcal{F}^0_t \right] - \mathbb{E}\left[ \frac{1}{N}\sum_{i=1}^N q^{*i}_s \Bigg| \mathcal{F}^0_s \right] \right\}\,ds,
\end{equation}
for the martingales $\left\{ M^{*i}_t, \mathcal{F}^{0i}_t\right\}_{0 \leqslant t \leqslant T},$ $i \in \mathcal{N},$ figuring in the FBSDE \eqref{eq:FBSDE}.
\end{theorem}


\begin{proof}
Throughout the present proof, we fix a competitive equilibrium $\left(p^N,\boldsymbol{\nu}^*,\Lambda^N\right).$ In particular, \eqref{eq:MktClearing1} holds, whence assumption \eqref{eq:Assumption2} implies $\xi^{i,N}_{t,u}\left( \boldsymbol{\nu}^* \right) = \xi^N_{t,u}\left(\Lambda^N\right),$ for $0\leqslant u \leqslant t \leqslant T$ and $i\in \mathcal{N}.$ 

From the inventory equation presented in \eqref{eq:FBSDE}, alongside \eqref{eq:MktClearing1}, we have
\begin{equation} \label{eq:OptimalAvgInventory}
    \mathbb{E}\left[ \frac{1}{N}\sum_{i=1}^N q^{*i}_t \Bigg| \mathcal{F}^0_t \right] = E^N_{0,t} - \int_0^t \Lambda^N_{u}\,du + \epsilon^N_t.
\end{equation}
Employing \eqref{eq:MktClearing1} and \eqref{eq:OptimalAvgInventory} in the average over $i \in \mathcal{N}$ of the BSDEs of \eqref{eq:FBSDE} projected on $\mathbb{F}^0,$ we deduce that 
\begin{align} \label{eq:PriceFormed1stStep}
    \begin{split}
        -2\kappa\Lambda^N_{t} &= - p^N_t +2\phi\int_0^t\left(E^N_{0,\tau} - \int_0^\tau \Lambda^N_{u}\,du + \epsilon^N_\tau \right)\,d\tau + \widetilde{f}^N_{t} + M^N_{t} + \frac{1}{N}\sum_{i=1}^N M^{*i}_0 \\
        &= - p^N_t +2\phi\left( \int_0^t E^N_{0,u}\,du - \int_0^t(t-u) \Lambda^N_{u}\,du + \int_0^t \epsilon^N_u \,du \right) + \widetilde{f}^N_{t} + M^N_{t} + \frac{1}{N}\sum_{i=1}^N M^{*i}_0,
    \end{split}
\end{align}
where we have written 
$$
\widetilde{f}^N_t = \mathbb{E}\left[ \xi^N_{T,t}\left(\Lambda_N\right) \frac{1}{N}\sum_{i=1}^N q^i_T - \int_t^T \xi_{u,t}^N\left(\Lambda_N\right) \frac{1}{N}\sum_{i=1}^N \nu^{*i}_u \,du \Bigg| \mathcal{F}^0_t \right].
$$
We claim $\widetilde{f}^N \equiv f^N,$ where we describe the latter in \eqref{eq:cStarN}. In effect, noticing that our assumption \eqref{eq:Assumption2} implies that the random variable $\xi^N_{s,t}\left(\Lambda^N\right)$ is $\mathcal{F}^0_t-$measurable, for $t \leqslant s \leqslant T,$ we use the conditional Fubini's Theorem, alongside the tower property of conditional expectations, to deduce the identities
\begin{align*}
    \widetilde{f}^N_t &= \xi^N_{T,t}\left(\Lambda_N\right) \mathbb{E}\left[ \frac{1}{N}\sum_{i=1}^N q^i_T \Bigg| \mathcal{F}^0_t \right] - \int_t^T \xi^N_{u,t}\left(\Lambda_N\right) \mathbb{E}\left[ \frac{1}{N}\sum_{i=1}^N \nu^{*i}_u \Bigg| \mathcal{F}^0_t \right] \,du \\
    &= \xi^N_{T,t}\left(\Lambda_N\right) \mathbb{E}\left[ \mathbb{E}\left[ \frac{1}{N}\sum_{i=1}^N q^i_T \Bigg| \mathcal{F}^0_T\right] \Bigg| \mathcal{F}^0_t \right] - \int_t^T \xi^N_{u,t}\left(\Lambda_N\right) \mathbb{E}\left[ \mathbb{E}\left[ \frac{1}{N}\sum_{i=1}^N \nu^{*i}_u \Bigg| \mathcal{F}^0_u \right] \Bigg| \mathcal{F}^0_t \right] \,du \\
    &= \xi^N_{T,t}\left(\Lambda_N\right) \mathbb{E}\left[ E^N_{0,T} - \int_0^T \Lambda^N_{u}\,du + \epsilon^N_T \Bigg| \mathcal{F}^0_t \right] + \int_t^T \xi^N_{u,t}\left(\Lambda_N\right) \mathbb{E}\left[ \Lambda^N_u \Bigg| \mathcal{F}^0_t \right] \,du = f^N_t,
\end{align*}
as we wanted to show. Rearranging \eqref{eq:PriceFormed1stStep}, we obtain
$$
p^N_t =  2\kappa\Lambda^N_{t} +2\phi\left(\int_0^t E^N_{0,u}\,du - \int_0^t(t-u) \Lambda^N_{u}\,du + \int_0^t \epsilon^N_u \,du \right) + f^N_{t} + M^N_t + \frac{1}{N}\sum_{i=1}^N M^{*i}_0,
$$
whence
$$
p^N_t - p^N_0 = 2\kappa\left( \Lambda^N_{t} - \Lambda^N_{0} \right) +2\phi\left( \int_0^t E^N_{0,u}\,du - \int_0^t(t-u) \Lambda^N_{u}\,du + \int_0^t \epsilon^N_u \,du \right) + F^N_{t} + M^N_{t}.
$$
The latter equation is clearly equivalent to \eqref{eq:FormedPrice}. This finishes the proof.
\end{proof}

Prior to our discussion of limiting consequences of Theorem \ref{thm:PriceFormed} as $N \to +\infty$, we explore some particular cases of this result in the two subsequent corollaries.  First, we consider the case in which the microstructural impact kernel, defined by $\xi_{t,u}^N$, does not depend on $u$.  Moreover,  this impact affects the price only through the aggregation of the $N$ players' turnover rate. Next,  we further assume that $\xi_t^N=\alpha/N$ and deduce that the formed price with this microstructural impact kernel leads to the permanent price impact model.


\begin{corollary} \label{cor:ParticularCase1}
Let us maintain all the notations of the statement of Theorem \ref{thm:PriceFormed}. If $\xi^N_{t,u}(\Lambda^N)$ is independent of $u,$ i.e., $\xi_{t,u}^N(\Lambda^N) = \xi_t^N(\Lambda^N),\, 0\leqslant u \leqslant t$, then we can write $p^N$ in \eqref{eq:FormedPrice} as follows:
\begin{equation} \label{eq:ParticularCase1}
    p^N_t = p^N_0 +2\phi\int_0^t E^N_{0,u}\,du - 2\phi \int_0^t (t-u) \Lambda^N_{u}\,du + 2\kappa\left(\Lambda^N_{t} - \Lambda^N_{0}\right) - \int_0^t \xi_u^N\left( \Lambda^N \right) \Lambda^N_{u}\,du + \int_0^t \epsilon^N_u\,du  + \widetilde{M}^N_{t},
\end{equation}
for a square-integrable martingale $\left\{ \widetilde{M}^N_t, \mathcal{F}^0_t \right\}_{0\leqslant t \leqslant T}.$
\end{corollary}
\begin{proof}
Under the current hypothesis, we derive from \eqref{eq:cStarN} that
$$
f^N_{t} = - \int_0^t \xi_u^N\left( \Lambda^N \right) \Lambda^N_{u}\,du + \widetilde{M}_t,
$$
where $\left\{\widetilde{M}_t,\mathcal{F}^0_t\right\}_{t}$ is a square-integrable martingale, cf. \eqref{eq:cStarN}. We can incorporate $\left\{ \widetilde{M}_t - \widetilde{M}_0,\mathcal{F}^0_t \right\}_{t}$ into $\left\{ M^N_{t},\mathcal{F}^0_t\right\}_{t}$ to obtain the representation we asserted for $p^N$ in \eqref{eq:ParticularCase1}.
\end{proof}

\begin{remark}
Notice that using \eqref{eq:ParticularCase1}, we can compute the quadratic variation of $p^N$ in terms of the quadratic variations of $\Lambda^N$ and $\widetilde{M}^N$. Indeed:
$$\langle p^N \rangle_t = 4 \kappa^2 \langle \Lambda^N \rangle_t + \langle \widetilde{M}^N \rangle_t + 2\kappa \langle \Lambda^N, \widetilde{M}^N \rangle_t.$$
One should compare this result with the one from MFG shown in Remark \ref{rmk:SDE_for_p_quadratic_var}. Notice that the finite player setting materializes into the quadratic variation of $\widetilde{M}^N$ and the co-variation of $\widetilde{M}^N$ and $\Lambda^N$.

\end{remark}

To state the second relevant particular case, we recall the following. Reference \cite{evangelista2020finite} contains the proof that in a linear price impacts model, with a slope of permanent price impact $\alpha,$ no inventory uncertainty ($\sigma^i = 0$), and with all the other parameters having the same meaning as in our current setting, the average of the inventories $E_N$ of traders is deterministic and satisfies the second-order ODE: 
\begin{equation} \label{eq:AvgInvsFinGame}
    \begin{cases}
        2\kappa E_N^{\prime\prime} + \alpha\left( 1- \frac{1}{N} \right) E_N^\prime - 2\phi E_N = 0,\\
        E_N(0) = E_{N,0} \text{ and } \kappa E^\prime_N(T) + A E_N(T) = 0,
    \end{cases}
\end{equation}
where $E_{N,0}$ is given.  Here, $E_N$ is defined as 
$$
E_N(t)=\frac{1}{N}\sum_{i=1}^N q^{*i}_t,
$$
where $q^{*i}$ is defined as in Proposition \ref{prop:CharacterizationOfTheNE}. This case corresponds to the finite population counterpart of the MFG model introduced by \cite{cardaliaguet2018mean}, in which the mean-field interactions are through the controls in a permanent (linear) price impact model. We can regard the next corollary as a sanity check. 

\begin{corollary} \label{cor:ParticularCase2}
Let us consider that all conditions we stipulated in Corollary \ref{cor:ParticularCase1} hold, and let us maintain its notation. If, moreover, $\xi_t^N \equiv \alpha/N,$ where $\alpha$ is a non-negative constant, $\Lambda^N_t \equiv -E_N^\prime(t),$ and $\sigma^i \equiv 0,$ for all $i,$ then
$$
p^N_t = p^N_0 + \alpha \left( E_{N}(t) - E_{N}(0) \right) + \widetilde{M}^N_{t}.
$$
\end{corollary}
\begin{proof}
Upon the use of \eqref{eq:AvgInvsFinGame}, we notice that the current assumptions imply
\begin{align*}
    p^N_t &= p_0^N + 2\kappa\left( E_N^\prime(0) - E_N^\prime(t) \right) +2\phi\int_0^t E_N(u)\,du + \frac{\alpha}{N}\left(E_N(t) - E_N(0)\right) + \widetilde{M}^N_{t}\\
    &= p_0^N + \alpha\left(1-\frac{1}{N}\right) \left(E_N(t) - E_N(0)\right)+ \frac{\alpha}{N}\left(E_N(t) - E_N(0)\right) + \int_0^t \epsilon_u\,du + \widetilde{M}^N_{t} \\
    &= p_0^N + \alpha \left(E_N(t) - E_N(0)\right) + \widetilde{M}^N_{t}.
\end{align*}
\end{proof}

\subsection{A convergence result}

The last result we present in this section comprises the conclusion of our initial goal. Namely, we provide conditions guaranteeing that the formed price in a competitive equilibrium (for the finite population game) converges to its MFG counterpart when $N$ goes to infinity.

\begin{theorem} \label{prop:Limit}
Let us consider a sequence $\left\{ \Lambda^N \right\}_N$ of elements of $\mathcal{A}_0.$ We keep all notations of the statement of Theorem \ref{thm:PriceFormed}. We suppose the following:
\begin{itemize}
    \item[(i)] The sequence of initial prices $\left\{ p^N_0\right\}_N$ satisfies $\mathbb{E}\left[ p^N_0 \right] \xrightarrow{N \rightarrow \infty} p_0,$ for some $p_0 \in \mathbb{R};$
    \item[(ii)] For some $E_0 \in \mathbb{R},$ we have $\mathbb{E}\left[ E^N_{0,0} \right] \xrightarrow{N\rightarrow \infty} E_0;$ 
    \item[(iii)] The sequence $\left\{ \Lambda^N \right\}_N$ is such that $\overline{\Lambda}(t) := \sup_{N\geqslant 1} \left|\mathbb{E}\left[ \Lambda^N_t \right] \right|$ satisfies $\int_0^T \overline{\Lambda}(t)\,dt < \infty;$
    \item[(iv)] For some $\Lambda : \left[0,T\right] \rightarrow \mathbb{R,}$ we have $\mathbb{E}\left[ \Lambda^{N}_t \right] \xrightarrow{N \rightarrow \infty} \Lambda(t),$ for each $t \in \left[0,T\right];$
    \item[(v)] The sequence $\left\{ F^N \right\}_N$ (cf. \eqref{eq:cStarN}) satisfies $\mathbb{E}\left[ F^N \right] \xrightarrow{N\rightarrow \infty} 0.$
\end{itemize}
Then, it follows that
\begin{equation} \label{eq:MainConvergence}
    \lim_{N \rightarrow \infty} \mathbb{E}\left[ p^N_t \right] = p(t),
\end{equation}
where $p$ is the formed price in the MFG with supply rate function $\Lambda,$ as in \eqref{eq:FormedPriceMFG}.
\end{theorem}
\begin{proof}
Taking expectations in \eqref{eq:FormedPrice}, we derive
$$
\mathbb{E}\left[ p^N_t \right] = \mathbb{E}\left[ p^N_0\right] + 2\kappa\left(\mathbb{E}\left[\Lambda^N_{t}\right] - \mathbb{E}\left[\Lambda^N_{0}\right]\right) +2\phi \mathbb{E}\left[ E^N_{0,0} \right] t - 2\phi \int_0^t (t-u) \mathbb{E}\left[ \Lambda^N_{u} \right]\,du.
$$
Using conditions $(i)-(v),$ with the aid of the Dominated Convergence Theorem, we can pass the above identity to the limit as $N \rightarrow \infty$ to conclude that
$$
\mathbb{E}\left[ p^N_t \right] \xrightarrow{N \rightarrow \infty} p_0 + 2\kappa\left(\Lambda(t) - \Lambda(0)\right) +2\phi E_0 t - 2\phi \int_0^t (t-u) \Lambda(t)\,du,
$$
which is readily equivalent to \eqref{eq:MainConvergence}.
\end{proof}

\begin{remark}\label{rem:condition}
Out of the conditions we stipulated in Theorem \ref{prop:Limit}, the one that probably deserves more attention is $(v).$ We observe from \eqref{eq:cStarN} that, if $\xi^N_{t,u}\left( \Lambda^N \right) = O\left( \frac{1}{N} \right),$ as $N\rightarrow \infty,$ then $(ii)-(iv),$ together with an appropriate integrability condition on $\left\{\epsilon^N\right\}_N,$ ensure that this condition does hold. Demanding this of $\xi^N$ is natural. Indeed, it follows if we assume $p^N(\boldsymbol{\nu}) = \widetilde{p}\left( -\mathbb{E}\left[ \frac{1}{N} \sum_{j=1}^N \nu^j_{\cdot} \Big| \mathcal{F}^0_{\cdot} \right] \right),$ for an adequate mapping $\Lambda \mapsto \widetilde{p}\left( \Lambda \right),$ cf. Example \ref{ex:PriceExample}. 

The general consideration about $(v)$ is that we must ask that the microstructural impact coefficient $\xi^N$ vanish, in an adequate sense, as the population size grows. We can relate it with opinion formation models in the economic literature, e.g., \cite{golub2010naive}. There, the authors characterize the phenomenon of the \textit{wisdom of the crowds}, in which all opinions in a large population converge to the truth, as being equivalent to the weight of the most influential agent vanishing as the population size grows. Our result seems in line with the intuitive idea that a vanishing microstructural influence leads, in the limit, to consistency (in a certain sense) to the MFG price. 
\end{remark}

\section{Empirical assessments} \label{sec:Empirical}

In this section we assess the performance of the formed price formulas under the MFG and finite population settings using real high-frequency data. We first specify the data and the order flow measures we consider. Next, we write down the regressions of price against explanatory variables that arise from the theory we developed in the previous sections. We finally show the numerical results of these experiments.

\subsection{The data} \label{subsec:data}


The data we use in this work can be found at \href{http://sebastian.statistics.utoronto.ca/books/algo-and-hf-trading/data/}{http://sebastian.statistics.utoronto.ca/books/algo-and-hf-trading/data/}.  It is from November 2014 for the following tickers: AMZN, EBAY, FB, GOOG, INTC, MSFT, MU, PCAR, SMH, VOD. This data contains all trades and snapshots of the LOB for every 100ms of the trading day.  Table \ref{tab:stocksData} displays the mean and standard deviation in parenthesis for each stock,  of the midprice, spread, and average daily volume (ADV).

\begin{footnotesize}
\begin{table}
\footnotesize
\begin{center}
\caption{The data we use in this work is from November 2014 for ten NASDAQ listed stocks.  It contains all trades and snapshots of the LOB for every 100ms of the trading day. }
\vspace{10pt}
\begin{tabular}{lllr}
\toprule
\textbf{Stocks} &  \textbf{Midprice} &   \textbf{Spread} &          \textbf{ADV} \\
\midrule
\rowcolor[gray]{0.9}    MU &      33.4 &     0.01 &    3,325,174 \\
\rowcolor[gray]{0.9}       &   (0.823) &  (0.001) &    (972,993) \\
\midrule
    FB &     74.68 &     0.01 &    5,616,503 \\
       &   (0.893) &  (0.002) &  (1,381,900) \\
\midrule       
\rowcolor[gray]{0.9}  PCAR &     66.57 &     0.02 &      373,722 \\
\rowcolor[gray]{0.9}       &    (0.64) &  (0.009) &    (122,522) \\
\midrule  
  AMZN &    317.18 &     0.11 &      706,172 \\
       &  (13.903) &  (0.049) &    (214,172) \\
\midrule    
\rowcolor[gray]{0.9}  EBAY &     54.11 &     0.01 &    1,224,661 \\
\rowcolor[gray]{0.9}       &   (0.664) &  (0.002) &    (262,809) \\
\midrule    
   SMH &     52.52 &     0.01 &      263,909 \\
       &   (1.011) &  (0.005) &     (95,574) \\
\midrule    
\rowcolor[gray]{0.9}  INTC &      34.5 &     0.01 &    4,717,047 \\
\rowcolor[gray]{0.9}       &   (1.104) &  (0.001) &  (2,069,603) \\
\midrule    
   VOD &     34.57 &     0.01 &      850,241 \\
       &   (1.196) &  (0.002) &    (588,596) \\
\midrule    
\rowcolor[gray]{0.9}  GOOG &    542.87 &     0.17 &      291,492 \\
\rowcolor[gray]{0.9}       &   (6.036) &   (0.08) &     (73,852) \\
\midrule    
  MSFT &      48.4 &     0.01 &    5,190,672 \\
             &   (0.704) &  (0.001) &    (971,123) \\
\bottomrule
\end{tabular}
\label{tab:stocksData}
\end{center}
\end{table}

\end{footnotesize}

\subsection{Order flow measures}

From the market clearing condition \eqref{eq:MktClearingMFG}, we see that $\Lambda$ should be taken as the additive opposite of a measure of order flow. Firstly, we consider the Trade Imbalance ($\TI$), see \cite{cont2014price}, as such a measure. We divide the trading day $\left[0,T\right]$ in $12$ corresponding intervals of $30$ minutes $\left[0=T_0,T_1\right],\ldots,\left[T_{11},T_{12}=360\right]$, further dividing such slice into smaller intervals $T_i = t_{i,0}<\ldots<t_{i,179} = T_{i+1}$ of constant length $t_{i,k+1}-t_{i,k}=10$ seconds. We then define the time series $\TI^i = \left\{ \TI^i_k \right\}_{k=0}^{179}$ on $\left[T_i,T_{i+1}\right]$ as
\begin{equation}\label{eq:tradeImbalance}
\TI^i_k = \sum_{n = N_{i,k} + 1}^{ N_{i,k+1} } \left\{ \mu^+_n - \mu^-_n \right\},  
\end{equation}
where $N_{i,k}$ denotes the number of events that have occurred on $\left[T_i,t_{i,k}\right],$ and $\mu^+_n$ (respectively, $\mu^-_n$) is the volume (possibly zero) of buying (respectively,  selling) market orders that are traded within $\left]t_{i,n-1},\,t_{i,n} \right],$ for each $n\geqslant 1.$

In order to take the LOB queue sizes into account, we also utilize a metric called Order Flow Imbalance (OFI), introduced in \cite[Subsection 2.1]{cont2014price}. The goal in using this is to consider information from the LOB first queues, since it comprises more information. We proceed to provide the definition of OFI, referring to their paper for an in-depth discussion on its design. Let us represent by $\fp^b_n$ and $q^b_n$ (respectively, $\fp^a_n$ and $q^a_n$) the bid (respectively, ask) price and volume (at-the-touch), respectively, resulting from the $n-$th set of events (i.e., those occurring within the window $\left]t_{i,n-1},\,t_{i,n}\right]$). We form the quantity
\begin{equation}\label{eq:enDef}
e_n := I_{\left\{ \fp^b_n \geqslant \fp^b_{n-1} \right\}}q^b_n -  I_{\left\{ \fp^b_n \leqslant \fp^b_{n-1} \right\}}q^b_{n-1} - I_{\left\{ \fp^a_n \leqslant \fp^a_{n-1} \right\}}q^a_n + I_{\left\{ \fp^a_n \geqslant \fp^a_{n-1} \right\}}q^a_{n-1},
\end{equation}
and we define
\begin{equation}\label{eq:OFIDef}
\OFI_k := \sum_{n = N_{i,k} + 1}^{N_{i,k+1}} e_n.
\end{equation}

\subsection{Regressions} \label{subsec:TI_MFG_Empirics}

Considering the formed price under the MFG setting \eqref{eq:FormedPriceMFG}, we undertake the subsequent regression over each time window $\left[T_i,T_{i+1}\right]$:
\begin{equation} \label{eq:RegressionFormedPrice_MFG}
    \fp^i_k = a_{i,0} S^{i,0}_k + a_{i,1} S^{i,1}_k + a_{i,2} S^{i,2}_k + a_{i,3} S^{i,3}_k + \epsilon^i_k,
\end{equation}
where we denote the error term by $\epsilon^i_k$, the observed midprice at time by $\fp^i_k$, both at time $t_{i,k}$, and the time series covariates:
\begin{equation} \label{eq:TimeSeries_MFG}
    \begin{cases}   
        S^{i,0} := \left\{S^{i,0}_k := 1 \right\}_{k = 0}^{179},\, S^{i,1} := \left\{S^{i,1}_k := t_{i,k} \right\}_{k = 0}^{179},\\
        S^{i,2} := \left\{S^{i,2}_k := \int_{T_i}^{t_{i,k}} \Lambda(u)(t_{i,k}-u)\,du \right\}_{k = 0}^{179},\, S^{i,3} := \left\{S^{i,3}_k := \Lambda(t_{i,k}) \right\}_{k = 0}^{179}.
    \end{cases}
\end{equation}
In $S^{i,2}$ we use the trapezoidal rule to approximate the integral. We assume parameters are constant in each of the time windows.

Furthermore, under the finite population setting, in particular considering the context and notation of Corollary \ref{cor:ParticularCase1}, we define the variable 
$$
I_{\text{micro}}(t) := \int_0^t \xi^N_u\left(\Lambda_u^N\right)\Lambda_u^N\,du.
$$
We notice that $I_{\text{micro}}$ figures in the price formation formula \eqref{eq:ParticularCase1} for the finite population game.  We call it $I_{\text{micro}}$ since it appears due to market friction due to the microstructural impact suffered by an individual trader.  If we assume, in particular, that this impact $\xi^N_u \equiv \xi^N$ is constant, then
$$
I_{\text{micro}}(t) = \xi^N \int_0^t \Lambda_u^N\,du.
$$
In the case of a decentralized game, it is reasonable to expect that $\xi^N = O(1/N),$ as we previously discussed in Remark \ref{rem:condition}. However, if there is some persistence in the impact of some ``special'' traders, as in a case when there is a hierarchy among the population, cf. \cite[Section 5]{evangelista2020finite}, we could expect that such term would play a relevant role in the price formation process.

Inspired by this discussion, we consider the regression \eqref{eq:RegressionFormedPrice_MFG} (over the time slices $[T_i,T_{i+1}]$),  adding the following explanatory variable:
\begin{align}\label{eq:s4}
S^{i,4} := \left\{S^{i,4}_k := \int_{T_i}^{t_{i,k}} \Lambda(u) \,du \right\}_{k=0}^{179},
\end{align}
for each window $\left[T_{i},T_{i+1}\right].$ In this way, the regression we investigate in the current setting is
\begin{equation} \label{eq:RegressionFinPop}
    \fp_k^i = a_{i,0} S_k^{i,0} + a_{i,1} S_k^{i,1} + a_{i,2} S_k^{i,2} + a_{i,3} S_k^{i,3} + a_{i,0} S_k^{i,4} + \eta_k^i,
\end{equation}
where we denote the error term by $\eta^i_k$ at time $t_{i,k}$. 

\subsection{Results}

We first consider that the order flow measure $\Lambda$ is given by the trade imbalance with $\Lambda(t_{i,k}) = -\TI^i_{k}.$ As a benchmark, we take a model inspired by \cite{kyle1985continuous}, viz.,
\begin{equation} \label{eq:RegressionTIBenchmark_MFG}
    \Delta \fp^i_k := \fp^i_k - \fp^i_{k-1} = b_{i,0} + b_{i,1} \TI^i_k + \widetilde{\epsilon}^i_k.
\end{equation}
Thus, in this benchmark, we simply assume that a price movement must correspond linearly to the market trade imbalance.

In Figure \ref{fig:ExampleOfTheFit}, we provide an example of the performance of our MFG model with parameters estimated as we described above, i.e., via the regression \eqref{eq:TimeSeries_MFG}, together with the benchmark model \eqref{eq:RegressionTIBenchmark_MFG}. In Table \ref{MFG_FinitePopR2_oneStock}, we present the average over all windows of of November, 3rd $2014$ for the GOOG stock.  In Figure \ref{fig:ExampleOfTheFit}, in green we have the formed price of the MFG using the TI order flow. The average $R^2$ on this trading day is given in the first column of Table \ref{MFG_FinitePopR2_oneStock} (cf. column $R^2_{TI}$). We also present the average $R^2$ over the whole period considered for all the stocks considered in this experiment in Table \ref{MFG_FinitePopR2_allStock}. We notice that the MFG model beats the benchmark model systematically over all stocks we considered. Moreover, we would point out that the $R^2$ achieved by the benchmark is comparable to the one observed in the literature; see \cite{cont2014price}, for instance.

To verify the statistical significance of our regression study, we display the Augmented Dickey–Fuller (ADF) test (cf. \cite{fuller1976introduction}) for the residuals  $\fp^i_k - a_{i,0} S^{i,0}_k + a_{i,1} S^{i,1}_k + a_{i,2} S^{i,2}_k + a_{i,3} S^{i,3}_k = \epsilon^i_k$.  The ADF statistic used in the test is a negative number. The more negative it is, the stronger the rejection of the hypothesis that there is a unit root at some level of confidence. The critical value of (5\%) for the sample size of 250 points is -3.43.  In Table \ref{tab:FormedPrice_TI}, the average ADF statistic is given under the column named ADF Resid..  Although the average result is not smaller than $-3.43$, we see that one standard deviation is enough to cover the critical level.  As we will see later in this section, the formed price under the finite population model 
displays better ADF tests' averages.

In Table \ref{tab:FormedPrice_TI},  we also showcase the average values over all windows and all days of the parameters we estimate from the regression \eqref{eq:RegressionTIBenchmark_MFG}. We immediately interpret the parameter $a_0$ as the average initial price of the corresponding stock sampled at the beginning of each window. 

\begin{table}
\centering
\caption{Average $R^2$ results of GOOG stock on November 3rd, 2014. }
\begin{tabular}{lrrrrrr}
\toprule
Stock &  $R^2_{TI}$ &  $R^2_{OFI}$ &  $R^2_{TI_{micro}}$ &  $R^2_{OFI_{micro}}$ &  $R^2_{TI_{bench}}$ &  $R^2_{OFI_{bench}}$ \\
\midrule
\rowcolor[gray]{0.9} GOOG &    0.505631 &     0.555541 &            0.703997 &             0.803111 &             0.17395 &             0.459976 \\
\bottomrule
\end{tabular}
\label{MFG_FinitePopR2_oneStock}
\end{table}

\begin{figure}[H]
\includegraphics[width=\textwidth]{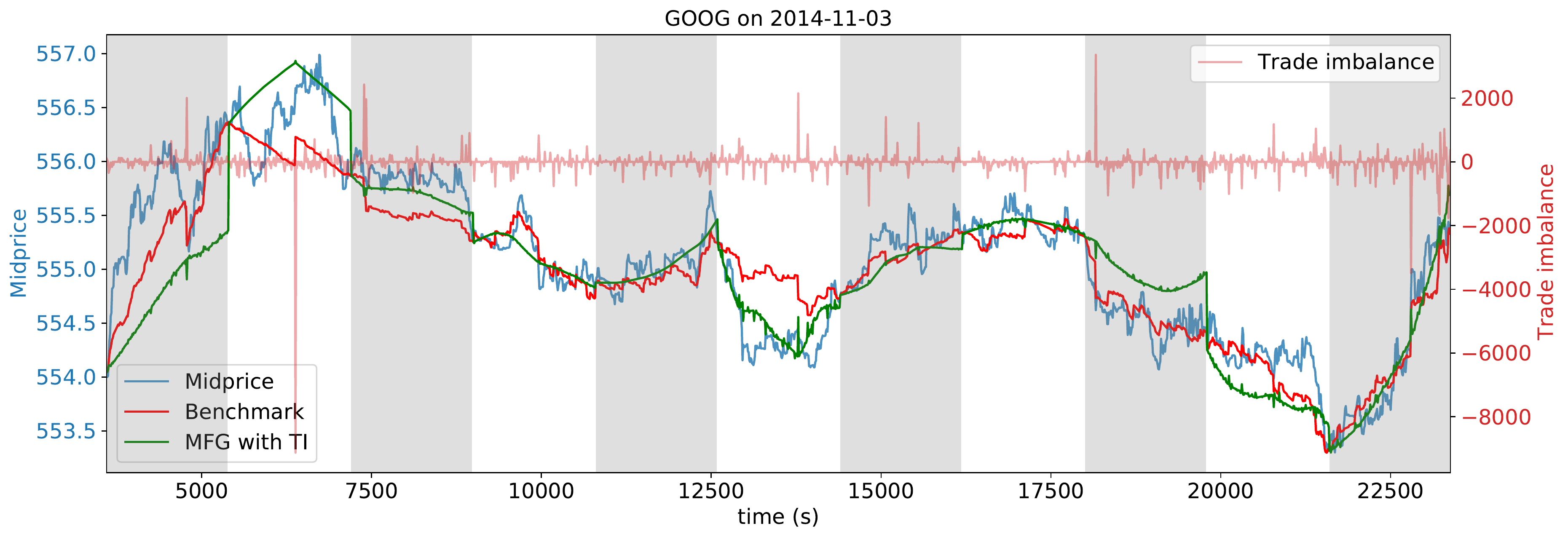}
\caption{
We use the Trade Imbalance, $\TI$ (see \cite{cont2014price}), as a measure of order flow (see \eqref{eq:tradeImbalance}) and the MFG model coefficients of the formed price equation: 
$
    \fp^i_k = a_{i,0} S^{i,0}_k + a_{i,1} S^{i,1}_k + a_{i,2} S^{i,2}_k + a_{i,3} S^{i,3}_k + \epsilon^i_k,
$
where we denote the error term by $\epsilon^i_k$ and the observed price by $\fp^i_k$, both at time $t_{i,k}$.  Such regression is performed in time intervals of $30$ minutes $\left[T_i,T_{i+1}\right]$ (marked as gray and white stripes), where model parameters are kept constant.  Moreover, we divide each window of 30 minutes into smaller intervals $T_i = t_{i,0}<\ldots<t_{i,179} = T_{i+1}$ of constant length $t_{i,k+1}-t_{i,k}=10$ seconds.  This regression corresponds to the MFG model in Theorem \ref{thm:MainTheorem}. As a benchmark, we take a model inspired by \cite{kyle1985continuous}, viz.,
$
    \Delta \fp^i_k := \fp^i_k - \fp^i_{k-1} = b_{i,0} + b_{i,1} \TI^i_k + \widetilde{\epsilon}^i_k.
$
This benchmark was also considered in \cite{cont2014price}. }
\label{fig:ExampleOfTheFit}
\end{figure}

Furthermore,  in Figure \ref{fig:FormedPriceTI_finPop} we present the outcome of considering the regression \eqref{eq:RegressionFinPop} from the finite population model compared to the trade imbalance linear model.  We show in Table \ref{MFG_FinitePopR2_oneStock} the average $R^2$ for the finite population model with the trade imbalance metric (cf. column $R^2_{TI_{micro}}$) on this trading day. For the $R^2$ over the whole analyzed period and for all the stocks, see Table \ref{MFG_FinitePopR2_allStock}. We observe that the extra explanatory variable of the finite population game brings relevant information to the problem. Indeed, this novel model not only beats the benchmark considerably buy it also improves the MFG model for all assets. The significance of the novel explanatory variable \eqref{eq:s4} is an indication that there might be a persistent effect of certain players in the market. The average values over all windows and all days of the regression parameters for this model is shown in Table \ref{tab:FormedPrice_TI_micro}.

\begin{figure}[H]
\includegraphics[width=\textwidth]{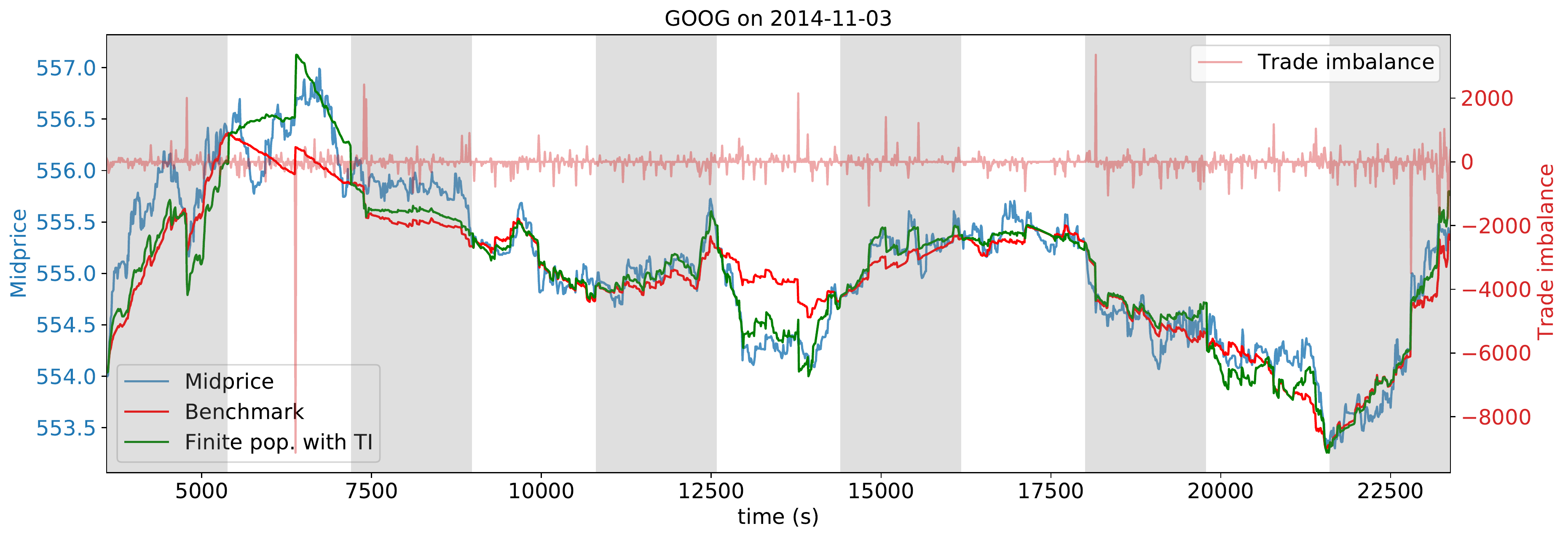}
\caption{We use the Trade Imbalance, $\TI$ (see \cite{cont2014price}), as a measure of order flow (see \eqref{eq:tradeImbalance}) and compute the finite population game model coefficients of the formed price equation: 
$
    \fp^i_k = a_{i,0} S^{i,0}_k + a_{i,1} S^{i,1}_k + a_{i,2} S^{i,2}_k + a_{i,3} S^{i,3}_k +a_{i,4} S^{i,4}_k + \eta^i_k,
$
where we denote the error term by $\eta^i_k$ and the observed price  by $\fp^i_k$, both at time $t_{i,k}$.  Such regression is performed in time intervals of $30$ minutes $\left[T_i,T_{i+1}\right],$ where model parameters are kept constant.  Moreover, we divide each window of 30 minutes into smaller intervals $T_i = t_{i,0}<\ldots<t_{i,179} = T_{i+1}$ of constant length $t_{i,k+1}-t_{i,k}=10$ seconds. This regression corresponds to the finite population price model in Corollary \ref{cor:ParticularCase1}.  As a benchmark, we take a model inspired by \cite{kyle1985continuous}, viz.,
$\Delta \fp^i_k := \fp^i_k - \fp^i_{k-1} = b_{i,0} + b_{i,1} \TI^i_k + \widetilde{\epsilon}^i_k.$
This benchmark was also considered in \cite{cont2014price}.}
\label{fig:FormedPriceTI_finPop}
\end{figure}

We now consider the order flow measure $\Lambda$ given by the Order Flow Imbalance with $\Lambda(t_{i,k}) = -\OFI^i_{k}.$ The benchmark we consider to compare our results for this metric is then
$$
\Delta \fp_k^i = c_{i,0} + c_{i,1} \OFI_k^i + \overline{\epsilon}^i_k,
$$
where the regression is performed over the slice $\left[T_i,T_{i+1}\right]$, which was studied in \cite{cont2014price}.

The counterpart of Figure \ref{fig:FormedPriceTI_finPop} in this setting is what we showcase in Figure \ref{fig:FormedPriceOFI_finPop}. This plot indicates that using this metric leads to overall better results than the trade imbalance one. This is expected since the OFI metric leads to better results when compared to TI, as noted by \cite{cont2014price}. A reason for the improvement can be due to the fact that more information is taken into account in the OFI metric, such as the size of the LOB queues.

From the results of Tables \ref{MFG_FinitePopR2_oneStock} and \ref{MFG_FinitePopR2_allStock}, we notice that using the OFI metric leads to substantial improvement of the performance of the current model. Moreover, the model we derived also outperforms the benchmark. This highlights the importance of the memory term in the price formation process. However, the memory term does not seem to be the key of the performance boost relative to the pure linear in the OFI metric. Indeed, from Tables \ref{MFG_FinitePopR2_oneStock} and \ref{MFG_FinitePopR2_allStock}, we see that when we employ $\Lambda = -\OFI,$ the average $R^2 $ of the MFG model is worse overall. Thus, the most valuable term from the perspective of the OFI metric seems to be the explanatory variable appearing from the microstructural impact, in the form of the time series \eqref{eq:s4}. Finally, the average values over all windows and all days of the regression parameters for this model is shown in Table \ref{tab:FormedPrice_OFI_micro}.

\begin{figure}[H]
\includegraphics[width=\textwidth]{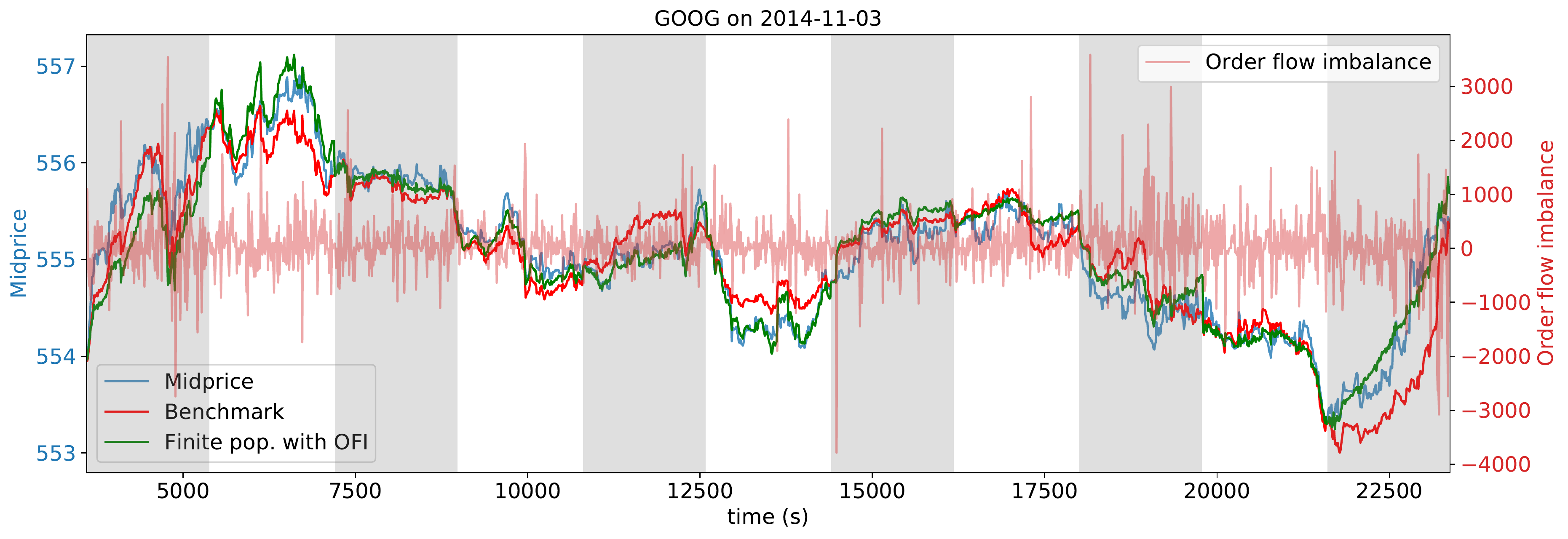}
\caption{We use the Oder Flow Imbalance, $\OFI$ (see \cite{cont2014price}), as a measure of order flow (see 
\eqref{eq:OFIDef}) and compute the finite population game model coefficients of the formed price equation: 
$
       \fp_k^i = a_{i,0} S_k^{i,0} + a_{i,1} S_k^{i,1} + a_{i,2} S_k^{i,2} + a_{i,3} S_k^{i,3} + a_{i,0} S_k^{i,4} + \eta_k^i,
$
where we denote the error term by $\eta^i_k$ and the observed price  by $\fp^i_k$, both at time $t_{i,k}$.  Such regression is performed in time intervals of $30$ minutes $\left[T_i,T_{i+1}\right],$ where model parameters are kept constant.  Moreover, we divide each window of 30 minutes into smaller intervals $T_i = t_{i,0}<\ldots<t_{i,179} = T_{i+1}$ of constant length $t_{i,k+1}-t_{i,k}=10$ seconds.  This regression corresponds to the finite population game price model in Corollary \ref{cor:ParticularCase1}.  As a benchmark, we take the model introduced in \cite{cont2014price}, viz. , $\Delta \fp_k^i = c_{i,0} + c_{i,1} \OFI_k^i + \overline{\epsilon}^i_k$, where $\OFI$ is given in \eqref{eq:OFIDef}.}
\label{fig:FormedPriceOFI_finPop}
\end{figure}

In Figure \ref{fig:allModels} we illustrate the performance of all models we developed: MFG with TI, MFG with OFI, Finite population with TI, and Finite Population with OFI. This figure corresponds to November 3rd, 2014, which is the same day illustrated in Figures \ref{fig:ExampleOfTheFit}, \ref{fig:FormedPriceTI_finPop}, and \ref{fig:FormedPriceOFI_finPop}.

To better visualize the differences of each formed price model, in Figure \ref{fig:allModels} we illustrate the squared  relative difference to the midprice, i.e., 
\begin{equation}\label{eq:sqrtRelDiff}
\bigg(\frac{p_t^{m} - p_t}{p_t}\bigg)^2,
\end{equation}
for $m \in \{$MFG with TI, MFG with OFI, Finite population with TI, and Finite Population with OFI$\}$, where $p_t$ is the midprice at time $t$. As for the other figures, the regression is performed in time intervals of $30$ minutes $\left[T_i,T_{i+1}\right],$ where model parameters are kept constant.  Moreover, we divide each window of 30 minutes into smaller intervals $T_i = t_{i,0}<\ldots<t_{i,179} = T_{i+1}$ of constant length $t_{i,k+1}-t_{i,k}=10$ seconds.  
Furthermore, in Table \ref{tab:sumSquaredRel} we compute the sum of the quantity in \eqref{eq:sqrtRelDiff} for the trading day of November 3rd, 2014, using the time discretization as mentioned above. 

In order to show the robustness of our results, we also considered in \ref{app-new-figs} two assets of very different nature than the ones considered so far. Namely, Petrobras ordinary stock and Bitcoin traded on Binance, both quoted in Brazilian real (BRL).


\begin{figure}[H]
\includegraphics[width=\textwidth]{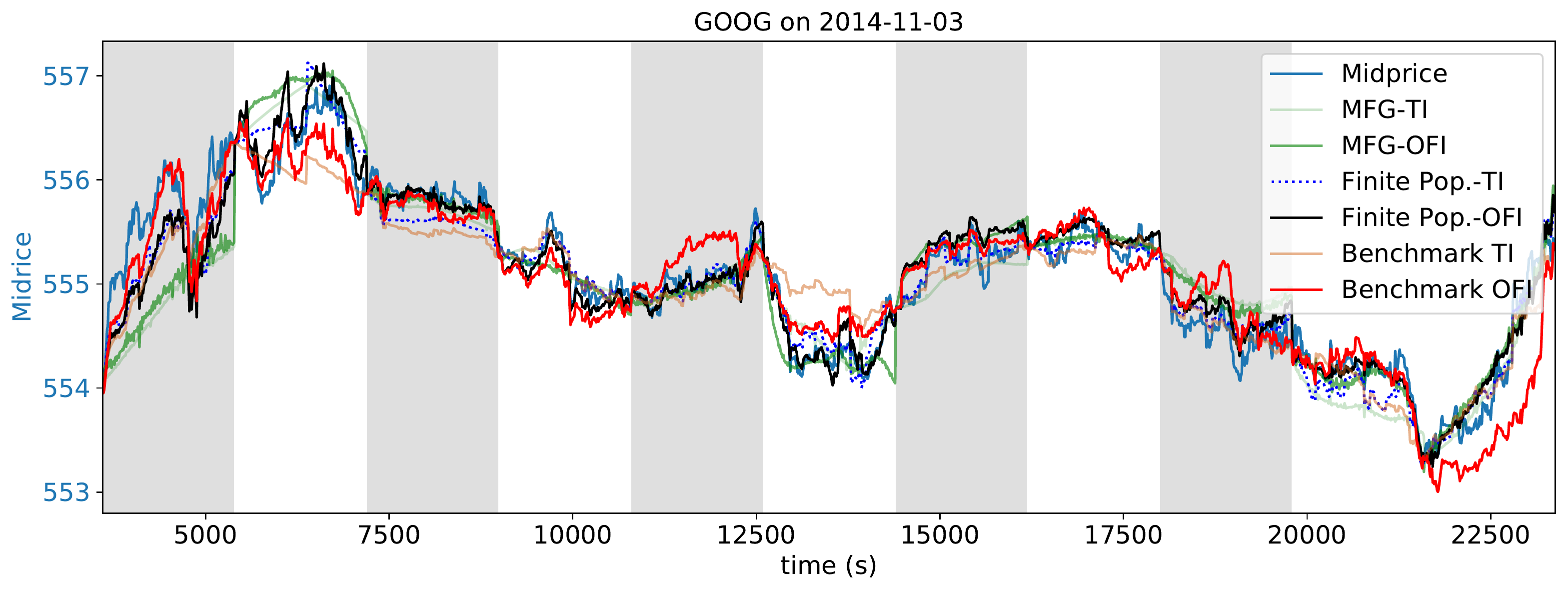}
\caption{Comparison between the midprice and formed price models: MFG with TI, MFG with OFI, Finite population with TI, and Finite population with OFI.}
\label{fig:allModels0}
\end{figure}

\begin{figure}[H]
\includegraphics[width=\textwidth]{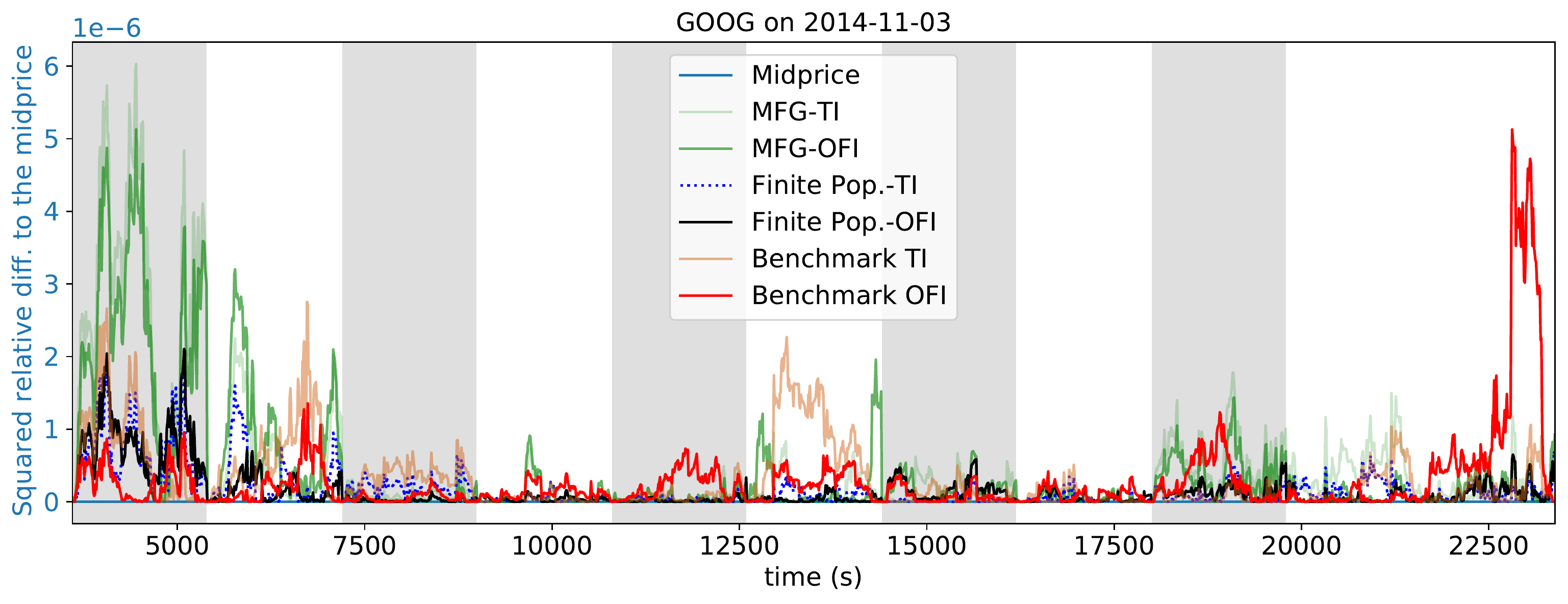}
\caption{Squared relative difference with respect to the midprice of each formed price models: MFG with TI, MFG with OFI, Finite population with TI,  and  Finite population with OFI.}
\label{fig:allModels}
\end{figure}

\begin{table}[H]
\centering
\caption{Mean of the squared relative difference with respect to the midprice, given by \eqref{eq:sqrtRelDiff}, for each model: MFG with TI, MFG with OFI, Finite population with TI, Finite Population with OFI.}
\begin{tabular}{lcccccc}
\toprule
Stock &  $\mbox{MFG}_{TI}$ &  $\mbox{MFG}_{OFI}$ &  $\mbox{Finite pop.}_{TI}$ &  $\mbox{Finite pop.}_{OFI}$ &  $\mbox{Bench}_{TI}$ &  $\mbox{Bench}-OFI$ \\
\midrule
\rowcolor[gray]{0.9} GOOG &    4.64e-7 &     4.05e-7 &            1.62e-7 &             1.45e-7 &             2.89e-7 &             2.72e-7 \\
\bottomrule
\end{tabular}
\label{tab:sumSquaredRel}
\end{table}


\section{Conclusions} \label{sec:conclusions}

We have proposed two frameworks of price formation. The first of which comprised a mean-field game model, where we determined the price assuming informed traders had objective functionals built such that they sought to maximize their P\&L, they did not want to terminate a given \added[id=F]{time} horizon holding any position and an appropriate market clearing condition was in force. Moreover, we considered a penalization on non-vanishing inventory.

The second framework we developed was a finite population counterpart of the former. Here, we could assume that prices were stochastic processes, relaxing the simplification we made in the MFG, where we assumed prices to be deterministic. We once again obtained a necessary condition determining prices in terms of the realized order flow, where we considered assumptions similar to those of the MFG for each player's utility in the finite population. We also showed, under appropriate assumptions, that our finite population result converged, on average, the MFG one. 

We provided empirical assessments for each of the models we proposed. In them, we considered two different types of order flow metrics,  which also highlight the importance of such a choice in the impact of our results. In conclusion, our best model was derived from the finite population model, using an order flow metric that took limit order book queue size information into account. A novel regression covariate in the finite population model showed to be significant, and it pointed out in the direction of the possible persistence of the influence of individual players in the price formation process. This last term suggests that one of the assumptions ensuring convergence of the finite population model to the MFG one is not valid in real market data, namely, that single individuals' contribution to order flow have little to no effect on the formed price.

\section*{Acknowledgments}

YT was financed in part by the Coordena\c{c}\~ao de Aperfei\c{c}oamento de Pessoal de N\'ivel Superior - Brasil (CAPES) - Finance Code 001.


\bibliographystyle{plainnat}
\bibliography{References}

\begin{thebibliography}{79}
\providecommand{\natexlab}[1]{#1}
\providecommand{\url}[1]{\texttt{#1}}
\expandafter\ifx\csname urlstyle\endcsname\relax
  \providecommand{\doi}[1]{doi: #1}\else
  \providecommand{\doi}{doi: \begingroup \urlstyle{rm}\Url}\fi

\bibitem[Almgren and Chriss(2001)]{almgren2001optimal}
Robert Almgren and Neil Chriss.
\newblock Optimal execution of portfolio transactions.
\newblock \emph{Journal of Risk}, 3:\penalty0 5--40, 2001.

\bibitem[Almgren and Li(2016)]{almgren2016option}
Robert Almgren and Tianhui~Michael Li.
\newblock Option hedging with smooth market impact.
\newblock \emph{Market microstructure and liquidity}, 2\penalty0 (01):\penalty0
  1650002, 2016.

\bibitem[Ashrafyan et~al.(2021)Ashrafyan, Bakaryan, Gomes, and
  Gutierrez]{ashrafyan2021duality}
Yuri Ashrafyan, Tigran Bakaryan, Diogo Gomes, and Julian Gutierrez.
\newblock A duality approach to a price formation mfg model.
\newblock \emph{arXiv preprint arXiv:2109.01791}, 2021.

\bibitem[Avellaneda and Stoikov(2008)]{avellaneda2008high}
Marco Avellaneda and Sasha Stoikov.
\newblock High-frequency trading in a limit order book.
\newblock \emph{Quantitative Finance}, 8\penalty0 (3):\penalty0 217--224, 2008.
\newblock URL
  \url{https://EconPapers.repec.org/RePEc:taf:quantf:v:8:y:2008:i:3:p:217-224}.

\bibitem[Bank et~al.(2017)Bank, Soner, and Vo{\ss}]{bank2017hedging}
Peter Bank, H~Mete Soner, and Moritz Vo{\ss}.
\newblock Hedging with temporary price impact.
\newblock \emph{Mathematics and financial economics}, 11\penalty0 (2):\penalty0
  215--239, 2017.

\bibitem[Bayraktar et~al.(2007)Bayraktar, Horst, and
  Sircar]{bayraktar2007queuing}
Erhan Bayraktar, Ulrich Horst, and Ronnie Sircar.
\newblock Queuing theoretic approaches to financial price fluctuations.
\newblock \emph{Handbooks in Operations Research and Management Science},
  15:\penalty0 637--677, 2007.

\bibitem[Bertsimas and Lo(1998)]{bertsimas1998optimal}
Dimitris Bertsimas and Andrew~W Lo.
\newblock Optimal control of execution costs.
\newblock \emph{Journal of Financial Markets}, 1\penalty0 (1):\penalty0 1--50,
  1998.

\bibitem[Biais et~al.(2005)Biais, Glosten, and Spatt]{biais2005market}
Bruno Biais, Larry Glosten, and Chester Spatt.
\newblock Market microstructure: A survey of microfoundations, empirical
  results, and policy implications.
\newblock \emph{Journal of Financial Markets}, 8\penalty0 (2):\penalty0
  217--264, 2005.

\bibitem[Booth et~al.(1999)Booth, So, and Tse]{booth1999price}
G~Geoffrey Booth, Raymond~W So, and Yiuman Tse.
\newblock Price discovery in the german equity index derivatives markets.
\newblock \emph{Journal of Futures Markets: Futures, Options, and Other
  Derivative Products}, 19\penalty0 (6):\penalty0 619--643, 1999.

\bibitem[Bouchaud(2009)]{bouchaud2009price}
Jean-Philippe Bouchaud.
\newblock Price impact.
\newblock \emph{arXiv preprint arXiv:0903.2428}, 2009.

\bibitem[Bouchaud et~al.(2002)Bouchaud, M{\'e}zard, Potters,
  et~al.]{bouchaud2002statistical}
Jean-Philippe Bouchaud, Marc M{\'e}zard, Marc Potters, et~al.
\newblock Statistical properties of stock order books: empirical results and
  models.
\newblock \emph{Quantitative finance}, 2\penalty0 (4):\penalty0 251--256, 2002.

\bibitem[Brennan and Subrahmanyam(1995)]{brennan1995investment}
Michael~J Brennan and Avanidhar Subrahmanyam.
\newblock Investment analysis and price formation in securities markets.
\newblock \emph{Journal of financial economics}, 38\penalty0 (3):\penalty0
  361--381, 1995.

\bibitem[Brennan and Subrahmanyam(1996)]{brennan1996market}
Michael~J Brennan and Avanidhar Subrahmanyam.
\newblock Market microstructure and asset pricing: On the compensation for
  illiquidity in stock returns.
\newblock \emph{Journal of financial economics}, 41\penalty0 (3):\penalty0
  441--464, 1996.

\bibitem[Cardaliaguet(2010)]{cardaliaguet2010notes}
Pierre Cardaliaguet.
\newblock Notes on mean field games.
\newblock Technical report, Technical report, 2010.

\bibitem[Cardaliaguet and Lehalle(2018)]{cardaliaguet2018mean}
Pierre Cardaliaguet and Charles-Albert Lehalle.
\newblock Mean field game of controls and an application to trade crowding.
\newblock \emph{Mathematics and Financial Economics}, 12\penalty0 (3):\penalty0
  335--363, 2018.

\bibitem[Carmona et~al.(2015)Carmona, Lacker, et~al.]{carmona2015probabilistic}
Ren{\'e} Carmona, Daniel Lacker, et~al.
\newblock A probabilistic weak formulation of mean field games and
  applications.
\newblock \emph{The Annals of Applied Probability}, 25\penalty0 (3):\penalty0
  1189--1231, 2015.

\bibitem[Cartea et~al.(2015)Cartea, Jaimungal, and
  Penalva]{cartea2015algorithmic}
{\'A}lvaro Cartea, Sebastian Jaimungal, and Jos{\'e} Penalva.
\newblock \emph{Algorithmic and high-frequency trading}.
\newblock Cambridge University Press, 2015.

\bibitem[Casgrain and Jaimungal(2020)]{casgrain2020mean}
Philippe Casgrain and Sebastian Jaimungal.
\newblock Mean-field games with differing beliefs for algorithmic trading.
\newblock \emph{Mathematical Finance}, 30\penalty0 (3):\penalty0 995--1034,
  2020.

\bibitem[Chakraborti et~al.(2011{\natexlab{a}})Chakraborti, Toke, Patriarca,
  and Abergel]{chakraborti2011econophysics}
Anirban Chakraborti, Ioane~Muni Toke, Marco Patriarca, and Fr{\'e}d{\'e}ric
  Abergel.
\newblock Econophysics review: Ii. agent-based models.
\newblock \emph{Quantitative Finance}, 11\penalty0 (7):\penalty0 1013--1041,
  2011{\natexlab{a}}.

\bibitem[Chakraborti et~al.(2011{\natexlab{b}})Chakraborti, Toke, Patriarca,
  and Abergel]{chakraborti2011econophysics1}
Anirban Chakraborti, Ioane~Muni Toke, Marco Patriarca, and Fr{\'e}d{\'e}ric
  Abergel.
\newblock Econophysics review: I. empirical facts.
\newblock \emph{Quantitative Finance}, 11\penalty0 (7):\penalty0 991--1012,
  2011{\natexlab{b}}.

\bibitem[Chu et~al.(1999)Chu, Hsieh, and Tse]{chu1999price}
Quentin~C Chu, Wen-liang~Gideon Hsieh, and Yiuman Tse.
\newblock Price discovery on the s\&p 500 index markets: An analysis of spot
  index, index futures, and spdrs.
\newblock \emph{International Review of Financial Analysis}, 8\penalty0
  (1):\penalty0 21--34, 1999.

\bibitem[Ciaian et~al.(2016)Ciaian, Rajcaniova, and Kancs]{ciaian2016economics}
Pavel Ciaian, Miroslava Rajcaniova, and d’Artis Kancs.
\newblock The economics of bitcoin price formation.
\newblock \emph{Applied Economics}, 48\penalty0 (19):\penalty0 1799--1815,
  2016.

\bibitem[Cohen et~al.(1981)Cohen, Maier, Schwartz, and
  Whitcomb]{cohen1981transaction}
Kalman~J Cohen, Steven~F Maier, Robert~A Schwartz, and David~K Whitcomb.
\newblock Transaction costs, order placement strategy, and existence of the
  bid-ask spread.
\newblock \emph{Journal of political economy}, 89\penalty0 (2):\penalty0
  287--305, 1981.

\bibitem[Cont et~al.(2010)Cont, Stoikov, and Talreja]{cont2010stochastic}
Rama Cont, Sasha Stoikov, and Rishi Talreja.
\newblock A stochastic model for order book dynamics.
\newblock \emph{Operations research}, 58\penalty0 (3):\penalty0 549--563, 2010.

\bibitem[Cont et~al.(2014)Cont, Kukanov, and Stoikov]{cont2014price}
Rama Cont, Arseniy Kukanov, and Sasha Stoikov.
\newblock The price impact of order book events.
\newblock \emph{Journal of financial econometrics}, 12\penalty0 (1):\penalty0
  47--88, 2014.

\bibitem[Daniels et~al.(2003)Daniels, Farmer, Gillemot, Iori, and
  Smith]{daniels2003quantitative}
Marcus~G Daniels, J~Doyne Farmer, L{\'a}szl{\'o} Gillemot, Giulia Iori, and
  Eric Smith.
\newblock Quantitative model of price diffusion and market friction based on
  trading as a mechanistic random process.
\newblock \emph{Physical review letters}, 90\penalty0 (10):\penalty0 108102,
  2003.

\bibitem[De~Jong and Schotman(2010)]{de2010price}
Frank De~Jong and Peter~C Schotman.
\newblock Price discovery in fragmented markets.
\newblock \emph{Journal of Financial Econometrics}, 8\penalty0 (1):\penalty0
  1--28, 2010.

\bibitem[Dias et~al.(2016)Dias, Fernandes, and Scherrer]{dias2016price}
Gustavo~F Dias, Marcelo Fernandes, and Cristina~M Scherrer.
\newblock Price discovery in a continuous-time setting.
\newblock \emph{Journal of Financial Econometrics}, 2016.

\bibitem[Dias et~al.(2018)Dias, Fernandes, and Scherrer]{dias2018price}
Gustavo~Fruet Dias, Marcelo Fernandes, and Cristina Scherrer.
\newblock Price discovery and market microstructure noise.
\newblock Technical report, Working paper, Sao Paulo School of Economics, 2018.

\bibitem[Duffie and Zhu(2017)]{duffie2017size}
Darrell Duffie and Haoxiang Zhu.
\newblock Size discovery.
\newblock \emph{The Review of Financial Studies}, 30\penalty0 (4):\penalty0
  1095--1150, 2017.

\bibitem[Easley and O'Hara(1995)]{easley1995market}
David Easley and Maureen O'Hara.
\newblock Market microstructure.
\newblock \emph{Handbooks in operations research and management science},
  9:\penalty0 357--383, 1995.

\bibitem[Evangelista and Thamsten(2020)]{evangelista2020finite}
David Evangelista and Yuri Thamsten.
\newblock Finite population games of optimal execution.
\newblock \emph{arXiv preprint arXiv:2004.00790}, 2020.

\bibitem[Fernandes and Scherrer(2018)]{fernandes2018price}
Marcelo Fernandes and Cristina~M Scherrer.
\newblock Price discovery in dual-class shares across multiple markets.
\newblock \emph{Journal of Futures Markets}, 38\penalty0 (1):\penalty0
  129--155, 2018.

\bibitem[F{\'e}ron et~al.(2020{\natexlab{a}})F{\'e}ron, Tankov, and
  Tinsi]{feron2020leader}
Olivier F{\'e}ron, Peter Tankov, and Laura Tinsi.
\newblock Price formation and optimal trading in intraday electricity markets
  with a major player.
\newblock \emph{arXiv preprint arXiv:2011.07655}, 2020{\natexlab{a}}.

\bibitem[F{\'e}ron et~al.(2020{\natexlab{b}})F{\'e}ron, Tankov, and
  Tinsi]{feron2020price}
Olivier F{\'e}ron, Peter Tankov, and Laura Tinsi.
\newblock Price formation and optimal trading in intraday electricity markets.
\newblock \emph{arXiv preprint arXiv:2009.04786}, 2020{\natexlab{b}}.

\bibitem[Fujii(2019)]{fujii2019probabilistic}
Masaaki Fujii.
\newblock Probabilistic approach to mean field games and mean field type
  control problems with multiple populations.
\newblock \emph{arXiv preprint arXiv:1911.11501}, 2019.

\bibitem[Fujii and Takahashi(2020{\natexlab{a}})]{fujii2020finite}
Masaaki Fujii and Akihiko Takahashi.
\newblock A finite agent equilibrium in an incomplete market and its strong
  convergence to the mean-field limit.
\newblock \emph{arXiv preprint arXiv:2010.09186}, 2020{\natexlab{a}}.

\bibitem[Fujii and Takahashi(2020{\natexlab{b}})]{fujii2020mean}
Masaaki Fujii and Akihiko Takahashi.
\newblock A mean field game approach to equilibrium pricing with market
  clearing condition.
\newblock \emph{arXiv preprint arXiv:2003.03035}, 2020{\natexlab{b}}.

\bibitem[Fuller(1976)]{fuller1976introduction}
Wayne~A Fuller.
\newblock Introduction to statistical time series, new york: Johnwiley.
\newblock \emph{FullerIntroduction to Statistical Time Series1976}, 1976.

\bibitem[Garbade and Silber(1979)]{garbade1979structural}
Kenneth~D Garbade and William~L Silber.
\newblock Structural organization of secondary markets: Clearing frequency,
  dealer activity and liquidity risk.
\newblock \emph{The Journal of Finance}, 34\penalty0 (3):\penalty0 577--593,
  1979.

\bibitem[Gatheral et~al.(2012)Gatheral, Schied, and
  Slynko]{gatheral2012transient}
Jim Gatheral, Alexander Schied, and Alla Slynko.
\newblock Transient linear price impact and fredholm integral equations.
\newblock \emph{Mathematical Finance: An International Journal of Mathematics,
  Statistics and Financial Economics}, 22\penalty0 (3):\penalty0 445--474,
  2012.

\bibitem[Gjerstad and Dickhaut(1998)]{gjerstad1998price}
Steven Gjerstad and John Dickhaut.
\newblock Price formation in double auctions.
\newblock \emph{Games and economic behavior}, 22\penalty0 (1):\penalty0 1--29,
  1998.

\bibitem[Golub and Jackson(2010)]{golub2010naive}
Benjamin Golub and Matthew~O Jackson.
\newblock Naive learning in social networks and the wisdom of crowds.
\newblock \emph{American Economic Journal: Microeconomics}, 2\penalty0
  (1):\penalty0 112--49, 2010.

\bibitem[Gomes et~al.(2021)Gomes, Gutierrez, and Ribeiro]{gomes2021random}
Diogo Gomes, Julian Gutierrez, and Ricardo Ribeiro.
\newblock A random-supply mean field game price model.
\newblock \emph{arXiv preprint arXiv:2109.01478}, 2021.

\bibitem[Gomes et~al.(2020)]{gomes2020mean}
Diogo~A Gomes et~al.
\newblock A mean-field game approach to price formation.
\newblock \emph{Dynamic Games and Applications}, pages 1--25, 2020.

\bibitem[Gonzalo and Granger(1995)]{gonzalo1995estimation}
Jesus Gonzalo and Clive Granger.
\newblock Estimation of common long-memory components in cointegrated systems.
\newblock \emph{Journal of Business \& Economic Statistics}, 13\penalty0
  (1):\penalty0 27--35, 1995.

\bibitem[Gonzalo and Ng(2001)]{gonzalo2001systematic}
Jesus Gonzalo and Serena Ng.
\newblock A systematic framework for analyzing the dynamic effects of permanent
  and transitory shocks.
\newblock \emph{Journal of Economic dynamics and Control}, 25\penalty0
  (10):\penalty0 1527--1546, 2001.

\bibitem[Grammig et~al.(2005)Grammig, Melvin, and
  Schlag]{grammig2005internationally}
Joachim Grammig, Michael Melvin, and Christian Schlag.
\newblock Internationally cross-listed stock prices during overlapping trading
  hours: price discovery and exchange rate effects.
\newblock \emph{Journal of Empirical Finance}, 12\penalty0 (1):\penalty0
  139--164, 2005.

\bibitem[Gu{\'e}ant(2016)]{gueant2016financial}
Olivier Gu{\'e}ant.
\newblock \emph{The Financial Mathematics of Market Liquidity: From optimal
  execution to market making}, volume~33.
\newblock CRC Press, 2016.

\bibitem[Gu{\'e}ant and Pu(2017)]{gueant2017option}
Olivier Gu{\'e}ant and Jiang Pu.
\newblock Option pricing and hedging with execution costs and market impact.
\newblock \emph{Mathematical Finance}, 27\penalty0 (3):\penalty0 803--831,
  2017.

\bibitem[Gu{\'e}ant et~al.(2015)Gu{\'e}ant, Pu, and
  Royer]{gueant2015accelerated}
Olivier Gu{\'e}ant, Jiang Pu, and Guillaume Royer.
\newblock Accelerated share repurchase: pricing and execution strategy.
\newblock \emph{International Journal of Theoretical and Applied Finance},
  18\penalty0 (03):\penalty0 1550019, 2015.

\bibitem[Hasbrouck(1995)]{hasbrouck1995one}
Joel Hasbrouck.
\newblock One security, many markets: Determining the contributions to price
  discovery.
\newblock \emph{The journal of Finance}, 50\penalty0 (4):\penalty0 1175--1199,
  1995.

\bibitem[Ho and Stoll(1980)]{ho1980dealer}
Thomas Ho and Hans~R Stoll.
\newblock On dealer markets under competition.
\newblock \emph{The Journal of Finance}, 35\penalty0 (2):\penalty0 259--267,
  1980.

\bibitem[Ho and Stoll(1981)]{ho1981optimal}
Thomas Ho and Hans~R Stoll.
\newblock Optimal dealer pricing under transactions and return uncertainty.
\newblock \emph{Journal of Financial economics}, 9\penalty0 (1):\penalty0
  47--73, 1981.

\bibitem[Ho and Stoll(1983)]{ho1983dynamics}
Thomas~SY Ho and Hans~R Stoll.
\newblock The dynamics of dealer markets under competition.
\newblock \emph{The Journal of finance}, 38\penalty0 (4):\penalty0 1053--1074,
  1983.

\bibitem[Huang et~al.(2006)Huang, Malham{\'e}, and Caines]{Caines1}
M.~Huang, R.~P. Malham{\'e}, and P.~E. Caines.
\newblock Large population stochastic dynamic games: closed-loop
  {M}c{K}ean-{V}lasov systems and the {N}ash certainty equivalence principle.
\newblock \emph{Commun. Inf. Syst.}, 6\penalty0 (3):\penalty0 221--251, 2006.
\newblock ISSN 1526-7555.
\newblock URL
  \url{http://projecteuclid.org/getRecord?id=euclid.cis/1183728987}.

\bibitem[Huang et~al.(2007)Huang, Caines, and Malham{\'e}]{Caines2}
M.~Huang, P.~E. Caines, and R.~P. Malham{\'e}.
\newblock Large-population cost-coupled {LQG} problems with nonuniform agents:
  individual-mass behavior and decentralized {$\epsilon$}-{N}ash equilibria.
\newblock \emph{IEEE Trans. Automat. Control}, 52\penalty0 (9):\penalty0
  1560--1571, 2007.
\newblock ISSN 0018-9286.
\newblock \doi{10.1109/TAC.2007.904450}.
\newblock URL \url{http://dx.doi.org/10.1109/TAC.2007.904450}.

\bibitem[Kallsen and Muhle-Karbe(2014)]{kallsen2014high}
Jan Kallsen and Johannes Muhle-Karbe.
\newblock High-resilience limits of block-shaped order books.
\newblock 2014.

\bibitem[Kyle(1985)]{kyle1985continuous}
Albert~S Kyle.
\newblock Continuous auctions and insider trading.
\newblock \emph{Econometrica: Journal of the Econometric Society}, 53\penalty0
  (6):\penalty0 1315--1335, 1985.
\newblock URL \url{http://dx.doi.org/10.2307/1913210}.

\bibitem[Lachapelle et~al.(2016)Lachapelle, Lasry, Lehalle, and
  Lions]{citeulike:12386824}
Aim\'{e} Lachapelle, Jean-Michel Lasry, Charles-Albert Lehalle, and
  Pierre-Louis Lions.
\newblock {Efficiency of the Price Formation Process in Presence of High
  Frequency Participants: a Mean Field Game analysis}.
\newblock \emph{Mathematics and Financial Economics}, 10\penalty0 (3):\penalty0
  223--262, June 2016.
\newblock URL \url{http://link.springer.com/article/10.1007/s11579-015-0157-1}.

\bibitem[Lan and Tan(2007)]{lan2007statistical}
Boon~Leong Lan and Ying~Oon Tan.
\newblock Statistical properties of stock market indices of different
  economies.
\newblock \emph{Physica A: Statistical Mechanics and its Applications},
  375\penalty0 (2):\penalty0 605--611, 2007.

\bibitem[Lasry and Lions(2006{\natexlab{a}})]{ll1}
J.-M. Lasry and P.-L. Lions.
\newblock Jeux \`a champ moyen. {I}. {L}e cas stationnaire.
\newblock \emph{C. R. Math. Acad. Sci. Paris}, 343\penalty0 (9):\penalty0
  619--625, 2006{\natexlab{a}}.
\newblock ISSN 1631-073X.

\bibitem[Lasry and Lions(2006{\natexlab{b}})]{ll2}
J.-M. Lasry and P.-L. Lions.
\newblock Jeux \`a champ moyen. {II}. {H}orizon fini et contr\^ole optimal.
\newblock \emph{C. R. Math. Acad. Sci. Paris}, 343\penalty0 (10):\penalty0
  679--684, 2006{\natexlab{b}}.
\newblock ISSN 1631-073X.

\bibitem[Lasry and Lions(2007)]{ll3}
J.-M. Lasry and P.-L. Lions.
\newblock Mean field games.
\newblock \emph{Jpn. J. Math.}, 2\penalty0 (1):\penalty0 229--260, 2007.
\newblock ISSN 0289-2316.

\bibitem[Lehalle(2013)]{lehalle2013market}
Charles-Albert Lehalle.
\newblock Market microstructure knowledge needed for controlling an intra-day
  trading process.
\newblock \emph{arXiv preprint arXiv:1302.4592}, 2013.

\bibitem[Lehalle et~al.(2011)Lehalle, Gu{\'e}ant, and
  Razafinimanana]{lehalle2011high}
Charles-Albert Lehalle, Olivier Gu{\'e}ant, and Julien Razafinimanana.
\newblock High-frequency simulations of an order book: a two-scale approach.
\newblock In \emph{Econophysics of Order-driven Markets}, pages 73--92.
  Springer, 2011.

\bibitem[Lien and Shrestha(2009)]{lien2009new}
Donald Lien and Keshab Shrestha.
\newblock A new information share measure.
\newblock \emph{Journal of Futures Markets: Futures, Options, and Other
  Derivative Products}, 29\penalty0 (4):\penalty0 377--395, 2009.

\bibitem[Locke and Onayev(2007)]{locke2007order}
Peter Locke and Zhan Onayev.
\newblock Order flow, dealer profitability, and price formation.
\newblock \emph{Journal of Financial Economics}, 85\penalty0 (3):\penalty0
  857--887, 2007.

\bibitem[Mastromatteo et~al.(2014)Mastromatteo, Toth, and
  Bouchaud]{mastromatteo2014agent}
Iacopo Mastromatteo, Bence Toth, and Jean-Philippe Bouchaud.
\newblock Agent-based models for latent liquidity and concave price impact.
\newblock \emph{Physical Review E}, 89\penalty0 (4):\penalty0 042805, 2014.

\bibitem[Mike and Farmer(2008)]{mike2008empirical}
Szabolcs Mike and J~Doyne Farmer.
\newblock An empirical behavioral model of liquidity and volatility.
\newblock \emph{Journal of Economic Dynamics and Control}, 32\penalty0
  (1):\penalty0 200--234, 2008.

\bibitem[Obizhaeva and Wang(2013)]{obizhaeva2013optimal}
Anna~A Obizhaeva and Jiang Wang.
\newblock Optimal trading strategy and supply/demand dynamics.
\newblock \emph{Journal of Financial Markets}, 16\penalty0 (1):\penalty0 1--32,
  2013.

\bibitem[Pham(2009)]{pham2009continuous}
Huy{\^e}n Pham.
\newblock \emph{Continuous-time stochastic control and optimization with
  financial applications}, volume~61.
\newblock Springer Science \& Business Media, 2009.

\bibitem[Potters and Bouchaud(2003)]{potters2003more}
Marc Potters and Jean-Philippe Bouchaud.
\newblock More statistical properties of order books and price impact.
\newblock \emph{Physica A: Statistical Mechanics and its Applications},
  324\penalty0 (1-2):\penalty0 133--140, 2003.

\bibitem[Ro{\c{s}}u(2009)]{rocsu2009dynamic}
Ioanid Ro{\c{s}}u.
\newblock A dynamic model of the limit order book.
\newblock \emph{The Review of Financial Studies}, 22\penalty0 (11):\penalty0
  4601--4641, 2009.

\bibitem[Schinckus(2012)]{schinckus2012methodological}
Christophe Schinckus.
\newblock Methodological comment on econophysics review i and ii: statistical
  econophysics and agent-based econophysics.
\newblock \emph{Quantitative Finance}, 12\penalty0 (8):\penalty0 1189--1192,
  2012.

\bibitem[Shrivats and Jaimungal(2020)]{shrivats2020optimal}
Arvind Shrivats and Sebastian Jaimungal.
\newblock Optimal generation and trading in solar renewable energy certificate
  (srec) markets.
\newblock \emph{Applied Mathematical Finance}, pages 1--33, 2020.

\bibitem[Shrivats et~al.(2020)Shrivats, Firoozi, and
  Jaimungal]{shrivats2020mean}
Arvind Shrivats, Dena Firoozi, and Sebastian Jaimungal.
\newblock A mean-field game approach to equilibrium pricing, optimal
  generation, and trading in solar renewable energy certificate (srec) markets.
\newblock \emph{arXiv preprint arXiv:2003.04938}, 2020.

\bibitem[Walras(1896)]{walras1896elements}
L{\'e}on Walras.
\newblock \emph{{\'E}l{\'e}ments d'{\'e}conomie politique pure, ou, Th{\'e}orie
  de la richesse sociale}.
\newblock F. Rouge, 1896.

\bibitem[Weber and Rosenow(2005)]{weber2005order}
Philipp Weber and Bernd Rosenow.
\newblock Order book approach to price impact.
\newblock \emph{Quantitative Finance}, 5\penalty0 (4):\penalty0 357--364, 2005.

\end{thebibliography}

\newpage

\appendix
\section{Tables} \label{app:table}


\setcounter{table}{0}

\begin{footnotesize}
\begin{table}[H]
\footnotesize
\centering
\caption{Average $R^2$ over the whole period analyzed under each model (including benchmarks), each measure of order flow and each stock studied in the paper.  We display the standard deviation of the computed $R^2$ in parenthesis. }  
\vspace{10pt}
\begin{tabular}{lcccccc}
\toprule
& \multicolumn{2}{c}{\textbf{MFG model}} & \multicolumn{2}{c}{\textbf{Finite pop. game model}} & \multicolumn{2}{c}{\textbf{Benchmarks} } 
\\
\midrule
Stock &   $R^2_{TI}$ &  $R^2_{OFI}$ & $R^2_{TI_{micro}}$ & $R^2_{OFI_{micro}}$ & $R^2_{TI_{bench}}$ & $R^2_{OFI_{bench}}$ \\
\midrule
\rowcolor[gray]{0.9}   MU &    0.5464292 &    0.5080833 &          0.7396288 &           0.9053048 &          0.1792843 &           0.6348812 \\
\rowcolor[gray]{0.9}      &  (0.2602291) &  (0.2782342) &        (0.2091207) &         (0.0867001) &         (0.100364) &         (0.1116312) \\
\midrule
   FB &    0.5627452 &    0.5499002 &          0.7927316 &           0.8879627 &          0.2700599 &            0.646802 \\
      &  (0.2679504) &  (0.2676806) &        (0.1691284) &         (0.1077574) &        (0.1129601) &         (0.1213636) \\
\midrule
\rowcolor[gray]{0.9} PCAR &     0.538717 &    0.5227527 &          0.7177695 &           0.8825702 &          0.1302854 &           0.5943707 \\
\rowcolor[gray]{0.9}      &  (0.2709382) &   (0.276842) &        (0.2070019) &          (0.122286) &        (0.0977697) &         (0.0919139) \\
\midrule
 AMZN &    0.5794828 &    0.5752787 &          0.7357521 &             0.81932 &          0.1862541 &            0.497083 \\
      &  (0.2529513) &  (0.2578789) &        (0.2089806) &         (0.1666376) &        (0.0902108) &         (0.1019693) \\
\midrule
\rowcolor[gray]{0.9} EBAY &    0.5287725 &    0.5161722 &          0.7524981 &           0.8985027 &          0.1888365 &           0.6237341 \\
\rowcolor[gray]{0.9}      &  (0.2692952) &   (0.276937) &        (0.1862722) &         (0.0981883) &        (0.1236811) &         (0.1201012) \\
\midrule
  SMH &    0.5671448 &    0.5273391 &          0.6386758 &           0.9045901 &          0.0383468 &           0.5366392 \\
      &  (0.2329803) &  (0.2483619) &        (0.2168032) &         (0.0835175) &        (0.0489217) &          (0.112495) \\
\midrule
\rowcolor[gray]{0.9} INTC &    0.5680237 &    0.5356633 &          0.7669449 &           0.9272422 &          0.1755975 &           0.6298803 \\
\rowcolor[gray]{0.9}      &  (0.2503211) &  (0.2681156) &        (0.1882361) &         (0.0654502) &        (0.1053752) &           (0.09795) \\
\midrule
  VOD &    0.5972025 &    0.5577175 &           0.751673 &            0.936199 &          0.0768654 &           0.5950157 \\
      &  (0.2682837) &  (0.2822306) &         (0.184862) &         (0.0649528) &        (0.0873495) &         (0.1315352) \\
\midrule
\rowcolor[gray]{0.9} GOOG &    0.5970308 &    0.5928375 &          0.7661087 &           0.8401094 &          0.1919996 &           0.4950987 \\
\rowcolor[gray]{0.9}      &  (0.2608654) &  (0.2629105) &        (0.1966308) &         (0.1517892) &        (0.1033645) &         (0.0956888) \\
\midrule
 MSFT &    0.5716435 &    0.5251484 &          0.7846665 &           0.9334051 &          0.2149702 &           0.6882648 \\
      &  (0.2491445) &  (0.2760493) &          (0.17097) &         (0.0655859) &        (0.1124921) &         (0.0949067) \\
\bottomrule
\end{tabular}
\label{MFG_FinitePopR2_allStock}
\end{table}

\end{footnotesize}

\begin{table}
\footnotesize
\begin{center}
\caption{This table displays the average model parameters and their standard deviation in parenthesis across all windows of $30$ min and all trading days of November 2014 for each stock. Here, we use the Trade Imbalance, $\TI$ (see \cite{cont2014price}), as a measure of order flow (see \eqref{eq:tradeImbalance}) and compute the MFG model coefficients of the formed price equation: 
$
    \fp^i_k = a_{i,0} S^{i,0}_k + a_{i,1} S^{i,1}_k + a_{i,2} S^{i,2}_k + a_{i,3} S^{i,3}_k + \epsilon^i_k,
$
where we denote the error term by $\epsilon^i_k$ and the observed price at time by $\fp^i_k$, both at time $t_{i,k}$.  Such regression is performed in time intervals of $30$ minutes $\left[T_i,T_{i+1}\right],$ where model parameters are kept constant.  Moreover, we divide each window of 30 minutes into smaller intervals $T_i = t_{i,0}<\ldots<t_{i,179} = T_{i+1}$ of constant length $t_{i,k+1}-t_{i,k}=10$ seconds.  This regression corresponds to the MFG model in Theorem \ref{thm:MainTheorem}. As a benchmark, we take a model inspired by \cite{kyle1985continuous}, viz.,
$
    \Delta \fp^i_k := \fp^i_k - \fp^i_{k-1} = b_{i,0} + b_{i,1} \TI^i_k + \widetilde{\epsilon}^i_k.
$
This benchmark was also considered in \cite{cont2014price}. }
\vspace{10pt}
\begin{tabular}{lcccccccc}
\toprule
      & \multicolumn{4}{c}{ \textbf{MFG model coeff.}} & & \multicolumn{2}{c}{ \textbf{Benchmark model coeff.}}\\
Stock &         $a_0$ &        $a_1$ &        $a_2$ &       $a_3$ & ADF Resid.  &     $b_0$ &     $b_1$ \\
\midrule
\rowcolor[gray]{0.9}   MU &    33.3963422 &    0.0851441 &     0.001033 &      -7e-07 &   -2.8765641 &    -4.99e-05 &       2e-06 \\
\rowcolor[gray]{0.9}      &    (0.845783) &  (0.5019527) &  (0.0107715) &   (1.3e-06) &  (0.8859914) &  (0.0004175) &   (1.1e-06) \\
\midrule
   FB &    74.6772701 &    0.0455335 &    0.0029012 &    -2.4e-06 &   -2.8127971 &   -0.0001932 &     3.5e-06 \\
      &   (0.9212966) &   (1.054565) &  (0.0143527) &   (1.9e-06) &  (0.8257791) &  (0.0007031) &   (1.8e-06) \\
\midrule
\rowcolor[gray]{0.9} PCAR &    66.5659824 &   -0.0138761 &   -0.0219083 &     1.9e-06 &   -2.8200954 &        8e-05 &    1.28e-05 \\
\rowcolor[gray]{0.9}      &   (0.6613602) &  (0.4460253) &  (0.1035331) &  (1.09e-05) &  (0.9382861) &  (0.0005088) &   (6.3e-06) \\
\midrule
 AMZN &   317.1961939 &   -0.0435624 &    0.0485113 &   -1.78e-05 &   -2.7000014 &    -5.12e-05 &    7.28e-05 \\
      &  (14.2859333) &  (4.0944619) &  (0.4366202) &  (4.78e-05) &  (0.8630214) &  (0.0043365) &  (3.32e-05) \\
\midrule
\rowcolor[gray]{0.9} EBAY &    54.1099383 &    0.0237025 &   -0.0026724 &      -6e-07 &   -2.7107647 &    -2.22e-05 &     4.1e-06 \\
\rowcolor[gray]{0.9}      &   (0.6804064) &  (0.3960183) &  (0.0219799) &   (2.5e-06) &  (0.9637356) &  (0.0004235) &   (2.4e-06) \\
\midrule
  SMH &    52.5159266 &     0.040239 &    -0.020156 &      -5e-07 &    -2.830251 &     7.09e-05 &       4e-06 \\
      &   (1.0403845) &  (0.4524628) &  (0.1338062) &  (1.19e-05) &  (0.7951221) &  (0.0004398) &   (6.7e-06) \\
\midrule
\rowcolor[gray]{0.9} INTC &    34.4902719 &    0.0760494 &    0.0003478 &      -3e-07 &   -2.9765764 &    -5.23e-05 &     1.2e-06 \\
\rowcolor[gray]{0.9}      &   (1.1351313) &  (0.4756311) &  (0.0059531) &     (7e-07) &  (0.8285871) &   (0.000403) &     (5e-07) \\
\midrule
  VOD &    34.5713667 &     0.002355 &   -0.0038229 &       8e-07 &   -2.7983246 &     2.91e-05 &     1.7e-06 \\
      &   (1.2332637) &  (0.2978866) &  (0.0250166) &   (2.3e-06) &  (0.8883831) &  (0.0002807) &   (1.6e-06) \\
\midrule
\rowcolor[gray]{0.9} GOOG &   542.8725774 &   -0.5931504 &     0.027602 &   -2.31e-05 &   -2.7499078 &   -0.0008275 &   0.0001529 \\
\rowcolor[gray]{0.9}      &   (6.1632073) &  (4.4657383) &    (0.87042) &  (0.000124) &  (0.9356923) &  (0.0049919) &  (8.38e-05) \\
\midrule
 MSFT &    48.3965698 &    0.0272846 &     0.000406 &      -2e-07 &   -2.8859986 &    -6.38e-05 &     1.3e-06 \\
      &   (0.7240534) &  (0.4672271) &  (0.0064598) &     (7e-07) &  (0.8539199) &  (0.0004826) &     (6e-07) \\
\bottomrule
\end{tabular}
\label{tab:FormedPrice_TI}
\end{center}
\end{table}


\begin{footnotesize}
\begin{table}
\footnotesize
\begin{center}
\caption{This table displays the average model parameters and their standard deviation in parenthesis across all windows of $30$ min and all trading days of November 2014 for each stock. Here, we use the Trade Imbalance $\TI$ (see \cite{cont2014price}) as a measure of order flow (see \eqref{eq:tradeImbalance}) and compute the finite population game model coefficients of the formed price equation: 
$
    \fp^i_k = a_{i,0} S^{i,0}_k + a_{i,1} S^{i,1}_k + a_{i,2} S^{i,2}_k + a_{i,3} S^{i,3}_k +a_{i,4} S^{i,4}_k \eta^i_k,
$
where we denote the error term by $\eta^i_k$ and the observed price at time by $\fp^i_k$, both at time $t_{i,k}$.  Such regression is performed in time intervals of $30$ minutes $\left[T_i,T_{i+1}\right],$ where model parameters are kept constant.  Moreover, we divide each window of 30 minutes into smaller intervals $T_i = t_{i,0}<\ldots<t_{i,179} = T_{i+1}$ of constant length $t_{i,k+1}-t_{i,k}=10$ seconds.  This regression corresponds to the finite population price model in Corollary \ref{cor:ParticularCase1}.  As a benchmark, we take a model inspired by \cite{kyle1985continuous}, viz.,
$\Delta \fp^i_k := \fp^i_k - \fp^i_{k-1} = b_{i,0} + b_{i,1} \TI^i_k + \widetilde{\epsilon}^i_k.$
This benchmark was also considered in \cite{cont2014price}.}
\vspace{10pt}
\begin{tabular}{lcccccccc}
\toprule
      & \multicolumn{5}{c}{ \textbf{Finite population model coeff.}} & & \multicolumn{2}{c}{ \textbf{Benchmark model coeff.}}\\
Stock &         $a_0$ &        $a_1$ &        $a_2$ &       $a_3$ &        $a_4$ &  ADF Resid. &     $b_0$ &     $b_1$ \\
\midrule
\rowcolor[gray]{0.9}   MU &    33.3976412 &   -0.0237521 &    0.0007491 &      -8e-07 &   -0.0004818 &   -3.5569321 &    -4.99e-05 &       2e-06 \\
\rowcolor[gray]{0.9}      &   (0.8461137) &  (0.3405124) &  (0.0095541) &     (1e-06) &  (0.0006709) &  (0.9644172) &  (0.0004175) &   (1.1e-06) \\
\midrule
   FB &    74.6859923 &    -0.148864 &    0.0033705 &    -2.5e-06 &   -0.0006395 &   -3.3782803 &   -0.0001932 &     3.5e-06 \\
      &     (0.91748) &  (0.6879874) &  (0.0115177) &   (1.5e-06) &    (0.00084) &  (0.9240582) &  (0.0007031) &   (1.8e-06) \\
\midrule
\rowcolor[gray]{0.9} PCAR &    66.5649544 &    0.0495253 &   -0.0189966 &     1.8e-06 &   -0.0037029 &   -3.4136948 &        8e-05 &    1.28e-05 \\
\rowcolor[gray]{0.9}      &   (0.6602739) &  (0.4140773) &  (0.0921828) &   (8.8e-06) &  (0.0049728) &  (0.9976919) &  (0.0005088) &   (6.3e-06) \\
\midrule
 AMZN &   317.2097833 &    -0.203643 &    0.0638642 &   -1.88e-05 &   -0.0151209 &   -3.1879314 &    -5.12e-05 &    7.28e-05 \\
      &  (14.3062892) &  (3.5137093) &  (0.4144437) &  (4.08e-05) &  (0.0225343) &  (0.8626377) &  (0.0043365) &  (3.32e-05) \\
\midrule
\rowcolor[gray]{0.9} EBAY &    54.1087666 &   -0.0507869 &   -0.0035469 &      -7e-07 &   -0.0011191 &   -3.4126401 &    -2.22e-05 &     4.1e-06 \\
\rowcolor[gray]{0.9}      &   (0.6797354) &  (0.3933079) &  (0.0173387) &   (2.1e-06) &  (0.0015897) &  (0.9577671) &  (0.0004235) &   (2.4e-06) \\
\midrule
  SMH &    52.5160044 &    0.0453901 &   -0.0122174 &        -0.0 &   -0.0010844 &   -3.1661944 &     7.09e-05 &       4e-06 \\
      &   (1.0403181) &  (0.4029302) &  (0.1250604) &  (1.04e-05) &   (0.004796) &  (0.7650656) &  (0.0004398) &   (6.7e-06) \\
\midrule
\rowcolor[gray]{0.9} INTC &    34.4926244 &   -0.0028086 &    -6.87e-05 &      -4e-07 &   -0.0002724 &    -3.626328 &    -5.23e-05 &     1.2e-06 \\
\rowcolor[gray]{0.9}      &   (1.1344234) &  (0.3030998) &  (0.0049006) &     (6e-07) &  (0.0003485) &  (0.9473067) &   (0.000403) &     (5e-07) \\
\midrule
  VOD &    34.5710566 &   -0.0225748 &   -0.0047916 &       6e-07 &   -0.0006803 &   -3.4371937 &     2.91e-05 &     1.7e-06 \\
      &   (1.2329606) &  (0.2898547) &  (0.0234281) &     (2e-06) &  (0.0010536) &  (0.9308903) &  (0.0002807) &   (1.6e-06) \\
\midrule
\rowcolor[gray]{0.9} GOOG &   542.8578395 &   -0.7487758 &    0.0149796 &      -3e-05 &   -0.0355172 &    -3.351813 &   -0.0008275 &   0.0001529 \\
\rowcolor[gray]{0.9}      &   (6.1634995) &  (3.4921034) &  (0.7157249) &  (9.66e-05) &  (0.0528001) &  (0.9196872) &  (0.0049919) &  (8.38e-05) \\
\midrule
 MSFT &    48.3974368 &   -0.0307244 &    0.0002427 &      -3e-07 &   -0.0003152 &   -3.5282415 &    -6.38e-05 &     1.3e-06 \\
      &   (0.7241274) &  (0.2945148) &  (0.0047474) &     (5e-07) &  (0.0004022) &  (0.9376129) &  (0.0004826) &     (6e-07) \\
\bottomrule
\end{tabular}
\label{tab:FormedPrice_TI_micro}
\end{center}
\end{table}

\end{footnotesize}
\clearpage

\begin{footnotesize}
\begin{table}
\footnotesize
\begin{center}
\caption{This table displays the average model parameters and their standard deviation in parenthesis across all windows of $30$ min and all trading days of November 2014 for each stock. Here, we use the Oder Flow Imbalance, $\OFI$ (see \cite{cont2014price}), as a measure of order flow (see 
\eqref{eq:OFIDef}) and compute the finite population game model coefficients of the formed price equation: 
$
       \fp_k^i = a_{i,0} S_k^{i,0} + a_{i,1} S_k^{i,1} + a_{i,2} S_k^{i,2} + a_{i,3} S_k^{i,3} + a_{i,0} S_k^{i,4} + \eta_k^i,
$
where we denote the error term by $\eta^i_k$ and the observed price at time by $\fp^i_k$, both at time $t_{i,k}$.  Such regression is performed in time intervals of $30$ minutes $\left[T_i,T_{i+1}\right],$ where model parameters are kept constant.  Moreover, we divide each window of 30 minutes into smaller intervals $T_i = t_{i,0}<\ldots<t_{i,179} = T_{i+1}$ of constant length $t_{i,k+1}-t_{i,k}=10$ seconds.  This regression corresponds to the finite population game price model in Corollary \ref{cor:ParticularCase1}.  As a benchmark, we take the model introduced in \cite{cont2014price}, viz. , $\Delta \fp_k^i = c_{i,0} + c_{i,1} \OFI_k^i + \overline{\epsilon}^i_k$, where $\OFI$ and $e_n$ are given in \eqref{eq:OFIDef} and \eqref{eq:enDef} respectively.}
\vspace{10pt}
\begin{tabular}{lcccccccc}
\toprule
      & \multicolumn{5}{c}{ \textbf{Finite population model coeff.}} & & \multicolumn{2}{c}{ \textbf{Benchmark model coeff.}}\\
Stock &         $a_0$ &        $a_1$ &        $a_2$ &       $a_3$ &        $a_4$ &  ADF Resid. &     $c_0$ &     $c_1$ \\
\midrule
\rowcolor[gray]{0.9}   MU &    33.3976715 &   -0.0061925 &    0.0001617 &      -8e-07 &   -0.0002785 &   -3.7419097 &    -4.06e-05 &     1.6e-06 \\
\rowcolor[gray]{0.9}      &   (0.8449511) &  (0.1907172) &  (0.0022751) &     (3e-07) &  (0.0003134) &  (1.2216516) &  (0.0002898) &     (5e-07) \\
\midrule
   FB &    74.6837965 &     0.053123 &     0.001143 &    -1.6e-06 &   -0.0004546 &     -3.36312 &    0.0001599 &     2.8e-06 \\
      &   (0.9150394) &  (0.5698996) &  (0.0037318) &     (7e-07) &  (0.0005303) &  (0.9568442) &  (0.0009176) &     (1e-06) \\
\midrule
\rowcolor[gray]{0.9} PCAR &    66.5654685 &   -0.0295009 &   -0.0021186 &    -6.5e-06 &    -0.001993 &   -3.4586142 &    -7.22e-05 &    1.15e-05 \\
\rowcolor[gray]{0.9}      &   (0.6580494) &  (0.3106973) &  (0.0200821) &   (2.6e-06) &  (0.0022566) &  (1.0904161) &  (0.0004809) &   (3.3e-06) \\
\midrule
 AMZN &   317.2187288 &    0.2560108 &   -0.0121746 &    -3.8e-05 &   -0.0111251 &     -3.15502 &    0.0018582 &    6.83e-05 \\
      &  (14.3286014) &   (3.204889) &  (0.2088567) &  (1.94e-05) &  (0.0131955) &  (0.8397061) &  (0.0046468) &  (1.88e-05) \\
\midrule
\rowcolor[gray]{0.9} EBAY &    54.1095513 &    0.0074091 &   -0.0003064 &    -1.5e-06 &   -0.0005658 &   -3.5466687 &     4.68e-05 &       3e-06 \\
\rowcolor[gray]{0.9}      &   (0.6779297) &   (0.211954) &  (0.0048564) &     (6e-07) &  (0.0006356) &  (0.9679349) &    (0.00035) &     (8e-07) \\
\midrule
  SMH &     52.513595 &    0.0264258 &    -6.32e-05 &      -6e-07 &   -0.0001945 &    -3.662269 &     -9.6e-06 &     1.1e-06 \\
      &    (1.042245) &   (0.213225) &  (0.0017799) &     (2e-07) &  (0.0002136) &  (1.0288569) &  (0.0003551) &     (3e-07) \\
\midrule
\rowcolor[gray]{0.9} INTC &    34.4921724 &    0.0070192 &    0.0001049 &      -4e-07 &     -0.00014 &   -3.8448383 &     1.25e-05 &       8e-07 \\
\rowcolor[gray]{0.9}      &   (1.1357877) &   (0.144242) &  (0.0010034) &     (1e-07) &  (0.0001523) &  (1.2681683) &  (0.0002294) &     (2e-07) \\
\midrule
  VOD &    34.5696753 &    0.0076454 &     -7.8e-06 &      -5e-07 &   -0.0002337 &   -4.2250149 &     1.14e-05 &     1.2e-06 \\
      &   (1.2324129) &  (0.1045205) &  (0.0017905) &     (2e-07) &  (0.0002649) &  (1.2164438) &  (0.0001424) &     (4e-07) \\
\midrule
\rowcolor[gray]{0.9} GOOG &   542.8302421 &     0.328999 &   -0.0009426 &   -4.72e-05 &   -0.0145369 &   -3.1616662 &    0.0006758 &    9.34e-05 \\
\rowcolor[gray]{0.9}      &   (6.1604382) &  (3.1658284) &   (0.235805) &   (2.7e-05) &   (0.017287) &  (0.9419513) &  (0.0055439) &  (2.64e-05) \\
\midrule
 MSFT &    48.3980857 &   -0.0037844 &    -1.15e-05 &      -4e-07 &   -0.0001448 &   -3.7722162 &    -4.43e-05 &       8e-07 \\
      &   (0.7234343) &   (0.171465) &  (0.0008954) &     (1e-07) &  (0.0001574) &  (1.0471774) &  (0.0002614) &     (2e-07) \\
\bottomrule
\end{tabular}
\label{tab:FormedPrice_OFI_micro}
\end{center}
\end{table}

\end{footnotesize}
\section{Proof of Proposition \ref{prop:CharacterizationOfTheNE}}\label{app:proof}

 \begin{proof}
 We take the G\^ateaux derivative,
 \begin{align*}
     \left\langle D_i J^i(\nu^i;\bm{\nu}^{-i}),w^i \right\rangle &= \lim_{\epsilon\rightarrow 0} \frac{J_i(\nu^{i} + \epsilon w^i,\bm{\nu}^{-i}) - J_i(\nu^{i},\bm{\nu}^{-i})}{\epsilon} \\
     &= \mathbb{E}\Bigg[ \int_0^T w^i_t \left\{p_T(\nu^i,\bm{\nu}^{-i}) - p_t(\nu^i,\bm{\nu}^{-i}) -2Aq^{i}_T - 2\kappa \nu^{i}_t-2\phi\int_t^T q^{i}_u\,du \right\}\,dt \\
     &\hspace{1.0cm}+ q^i_T\left\langle D_ip_T(\nu^i,\bm{\nu}^{-i}),w^i \right\rangle - \int_0^T \nu^i_t \left\langle D_ip_t(\nu^i,\bm{\nu}^{-i}),w^i\right\rangle\,dt\Bigg].
 \end{align*}
 By our assumptions (\textbf{A1}) and (\textbf{A2}) this leads to 
 \begin{align} \label{eq:GateauxDerivative}
     \begin{split}
         \left\langle D_i J^i(\nu^i,\bm{\nu}^{-i}),w^i \right\rangle = \mathbb{E}\Bigg[ \int_0^T w^i_t \Bigg\{&p_T(\nu^i,\bm{\nu}^{-i})- p_t(\nu^i,\bm{\nu}^{-i}) -(2A - \xi^{i,N}_{T,t}(\nu^i,\bm{\nu}^{-i}))q^{i}_T \\
         &- 2\kappa \nu^{i}_t-2\phi\int_t^T q^{i}_u\,du - \int_t^T \nu^i_u \xi_{u,t}^i(\nu^i,\bm{\nu}^{-i})\,du \Bigg\}\,dt \Bigg]\\
         = \mathbb{E}\Bigg[\int_0^T w^i_t\Bigg\{& -2\kappa \nu^i_t -p_t\left(\nu^i,\bm{\nu}^{-i}\right) + 2\phi \int_0^t q^i_u\,du \\
         & +\mathbb{E}\left[\xi^i_{T,t}(\nu^i,\bm{\nu}^{-i})q^i_T - \int_t^T\nu^i_u \xi_{u,t}^i(\nu^i,\bm{\nu}^{-i})\,du \Bigg| \mathcal{F}^{i0}_t \right] + M^i_t  \Bigg\}\,dt \Bigg], 
     \end{split} 
 \end{align} 
 where we have written
 $$
 M^i_t := \mathbb{E}\left[-2A q^i_T + p_T(\nu^i,\bm{\nu}^{-i}) - 2\phi \int_0^T q^i_u\,du \Bigg| \mathcal{F}^{i0}_t\right] \in \mathbb{L}^2.
 $$
 A Nash equilibrium $\bm{\nu}^*$ is characterized by
 $$
 \nu^{*i} = \argmax_{\nu^i} J_i(\nu^i,\bm{\nu}^{*-i}), \text{ for each } i \in \mathcal{N}.
 $$
 Under the assumptions presently made, this is equivalent to
 \begin{equation} \label{eq:CharacterizationViaGateauxDerivatives}
     \left\langle D_iJ_i(\nu^{*i}, \bm{\nu}^{*-i}),w^i \right\rangle = 0,
 \end{equation}
 for all $i \in \mathcal{N}$ and $w^i.$ From (\ref{eq:GateauxDerivative}) and (\ref{eq:CharacterizationViaGateauxDerivatives}), the characterization of $\bm{\nu}^*$ as the solution of (\ref{eq:FBSDE}) follows.
 \end{proof}

\section{Further discussion on the admissible prices class} \label{app:discAdmPrices}

We begin by remarking that we can generalize Example \ref{ex:PriceExample} as follows.

\begin{example} \label{ex:PriceExample_app}
Both conditions (\textbf{C1}) and (\textbf{C2}) hold if $p_t(\bm{\nu}) = P_t\left(\mathbb{E}\left[ \frac{1}{N}\sum_{i=1}^N q^i_t \Big| \mathcal{F}^0_t \right] \right),$ where $(t,Q) \in \left[0,T\right] \times \mathbb{R} \mapsto P_t(Q) \in \mathbb{R}$ satisfies:
\begin{itemize}
    \item[(i)] For each $\left\{ Q_t \right\}_{0\leqslant t \leqslant T} \in \mathcal{A}_0,$ the process $\{P_t(Q_t)\}_{0\leqslant t \leqslant T}$ belongs to $\mathcal{A}_0.$
    \item[(ii)] For each $t \in \left[0,T\right],$ the mapping $Q \mapsto P_t(Q)$ is of class $C^1,$ $\mathbb{P}-$almost surely.
\end{itemize}
In effect, in this setting we have
$$
\left\langle D_ip_t(\bm{\nu}),w^i \right\rangle = \mathbb{E}\left[ \int_0^t \frac{1}{N}\partial_Q P_t\left(\mathbb{E}\left[ \frac{1}{N}\sum_{i=1}^N q^i_t \Bigg| \mathcal{F}^0_t \right] \right)w^i_u\,du \Bigg| \mathcal{F}^0_t \right],
$$
whence (\textbf{C1}) is valid with $\xi^{i,N}_{t,u}(\boldsymbol{\nu}) = \frac{1}{N}\partial_Q P_t\left(\mathbb{E}\left[ \frac{1}{N}\sum_{i=1}^N q^i_t \Bigg| \mathcal{F}^0_t \right] \right),$ $0 \leqslant u \leqslant t \leqslant T,$ from where it follows that (\textbf{C2}) holds too.
\end{example}

The works \cite{feron2020price,feron2020leader} focus on the analysis of the price formation problem in the case information is symmetric among players and $p$ is has the form we described in \eqref{eq:examplePrice}. In this case, understanding the formed price amounts to understanding the dynamics of the optimal trading rates. In the present framework, in which we allow for a broader class of admissible prices, we gain more flexibility for explaining the price in terms of the order flow, as we have already seen in Section \ref{sec:MFG}.

We remark that (\textbf{C2}) is an integrability condition we stipulate to allow us to apply variational techniques. Regarding the condition (\textbf{C1}), it is of a rather general nature. The intuition behind it is that a mapping $p \in \mathfrak{P}$ ought to have the following properties: 
\begin{itemize}
    \item[(i)] Prices are \textit{non-anticipative} functionals: for each $t \in \left[0,T\right],$ $\boldsymbol{\nu} \mapsto p_t\left(\boldsymbol{\nu}\right)$ depends only on the restriction $\boldsymbol{\nu}|_{\left[0,t\right]}$ of the path of $\boldsymbol{\nu};$
    \item[(ii)] For each $t \in \left[0,T\right],$ the mapping
                $$
                \boldsymbol{\nu}(\omega)|_{\left[0,t\right]} = \left\{\boldsymbol{\nu}_u(\omega) \right\}_{0\leqslant u\leqslant t} \in L^2(0,t)^N \mapsto p_t\left(\boldsymbol{\nu}\right)(\omega) \in \mathbb{R}
                $$
                is G\^ateaux differentiable $\mathbb{P}-$almost surely;
    \item[(iii)] For every $\boldsymbol{\nu} \in \mathcal{A},$ the process $\left\{ p_t\left(\boldsymbol{\nu}\right) \right\}_{0\leqslant t \leqslant T}$ belongs to $\mathcal{A}_0.$
\end{itemize}
Technically, it is worthwhile to remark an implication of the relation $p \in \mathfrak{P}.$ We expose it in the sequel.

\begin{lemma} \label{lem:ClassPrices}
Let us assume $p \in \mathfrak{P}.$ Then,
$$
p_t\left( \boldsymbol{\nu} \right) = p_t(\boldsymbol{0}) + \mathbb{E}\left[ \sum_{i=1}^N \int_0^t \zeta^i_{t,u}\left( \boldsymbol{\nu} \right) \nu^i_u\,du \Bigg| \mathcal{F}^0_t \right],
$$
where
$$
\zeta^i_{t,u}\left( \boldsymbol{\nu} \right) = \int_0^1 \xi^i_{t,u}\left( 0,\ldots,0,\lambda \nu^{i}, \nu^{i+1},\ldots,\nu^N \right) \,d\lambda.
$$
\end{lemma}
\begin{proof}
It suffices to observe that
\begin{align*}
    p_t(\boldsymbol{\nu}) - p_t(\boldsymbol{0}) &= \sum_{i=1}^N \left\{ p_t(0,\ldots,0,\nu^i,\nu^{i+1},\ldots,\nu^N) - p_t(0,\ldots,0,0,\nu^{i+1},\ldots,\nu^N) \right\}  \\
    &=\int_0^1 \frac{d}{d\lambda} p_t(0,\ldots,0,\lambda \nu^{i},\nu^{i+1},\ldots,\nu^N)\,d\lambda \\
    &= \int_0^1 \sum_{i=1}^N \mathbb{E}\left[ \int_0^t \xi^{i,N}_{t,u}\left(0,\ldots,0,\lambda \nu^{i},\nu^{i+1},\ldots,\nu^N\right)\nu^i_u\,du \Bigg| \mathcal{F}^0_t \right] d\lambda\\
    &= \mathbb{E}\left[ \sum_{i=1}^N \int_0^t \left( \int_0^1 \xi^{i,N}_{t,u}\left(0,\ldots,0,\lambda \nu^{i},\nu^{i+1},\ldots,\nu^N\right) \,d\lambda \right) \nu^i_u\,du \Bigg| \mathcal{F}^0_t \right].
\end{align*}
\end{proof}

Lemma \ref{lem:ClassPrices} indicates that our class $\mathfrak{P}$ generalizes in a natural way the concept of transient price impact \cite{gatheral2012transient}.

\section{Examples from the Brazilian and Crypto markets}\label{app-new-figs}

Here we illustrate two examples of our price formation approach applied to the Brazilian market (PETR3, Petrobras ordinary stock) and crypto market (BTCBRL, Bitcoin traded on Binance exchange in Brazil), quoted in BRL.  For both of them,  we illustrate the formed prices from the four methods we developed (MFG with TI, MFG with OFI, Finite population with TI, and finite population with OFI), along with the corresponding benchmarks.  Figures \ref{fig:allModelsPETR3} and \ref{fig:allModelsPETR32} correspond to PETR3, and in Table \ref{tab:R2-PETR3} we compute the average of adjusted $R^2$ for this stock.  Moreover,  Figures \ref{fig:allModelsBTCBRL} and \ref{fig:allModelsBTCBRL2} relate to BTCBRL, and in Table \ref{tab:R2-BTCBRL} we compute the average of its adjusted $R^2$.
\begin{figure}[H]
\includegraphics[width=\textwidth]{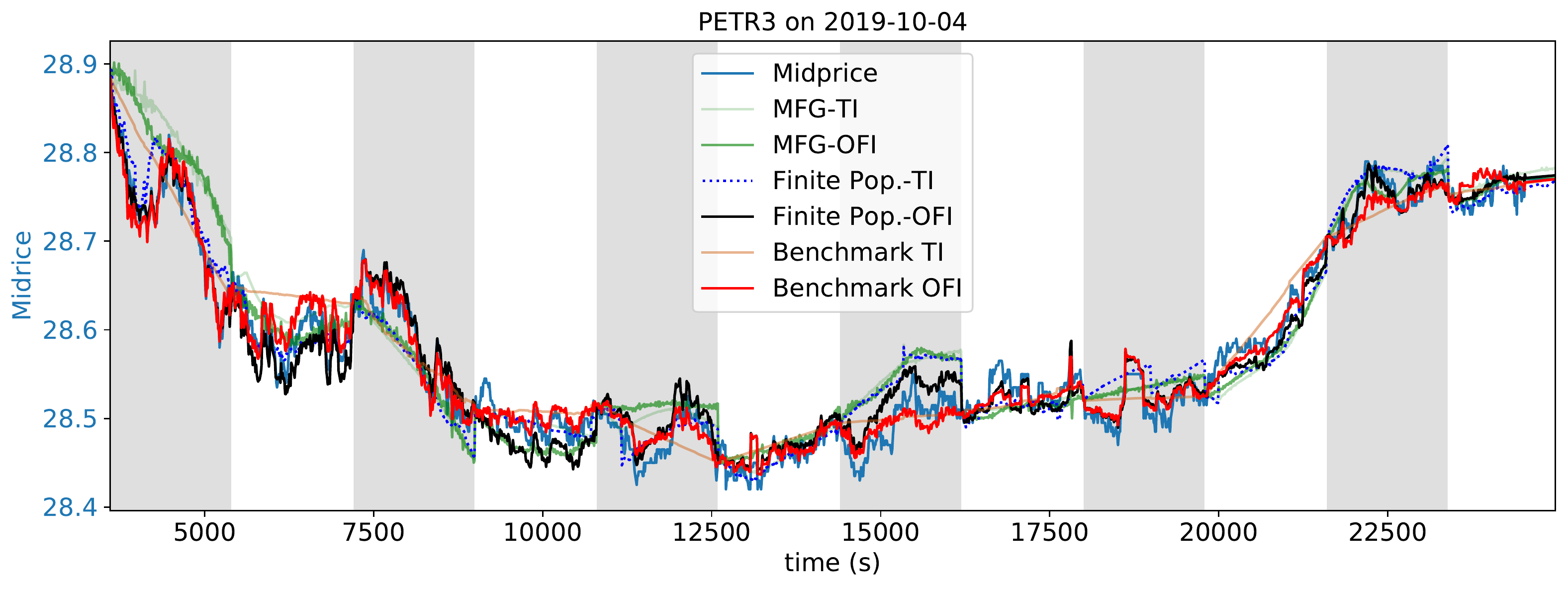}
\caption{Comparison between the midprice and formed price models: MFG with TI, MFG with OFI, Finite population with TI, and Finite population with OFI.}
\label{fig:allModelsPETR3}
\end{figure}

\begin{figure}[H]
\includegraphics[width=\textwidth]{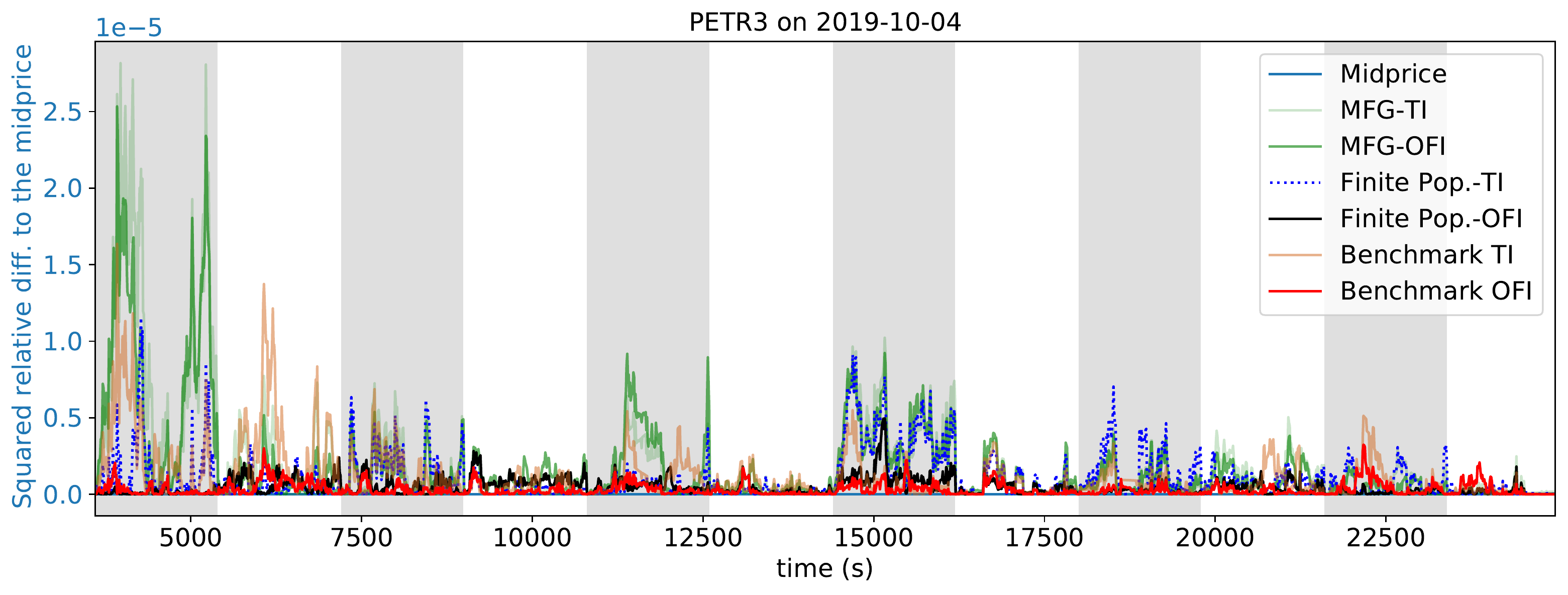}
\caption{Squared relative difference with respect to the midprice of each formed price models: MFG with TI, MFG with OFI, Finite population with TI,  and  Finite population with OFI.}
\label{fig:allModelsPETR32}
\end{figure}
\begin{table}
\centering
\caption{Average adjusted $R^2$ results}
\begin{tabular}{lcccccc}
\toprule
 Stock &  $R^2_{TI}$ &  $R^2_{OFI}$ &  $R^2_{TI_{micro}}$ &  $R^2_{OFI_{micro}}$ &  $R^2_{TI_{bench}}$ &  $R^2_{OFI_{bench}}$ \\
\midrule
\rowcolor[gray]{0.9} PETR3 &    0.486544 &      0.51164 &             0.59331 &             0.848529 &            0.004201 &             0.460331 \\
\bottomrule
\end{tabular}
\label{tab:R2-PETR3}
\end{table}

\begin{figure}[H]
\includegraphics[width=\textwidth]{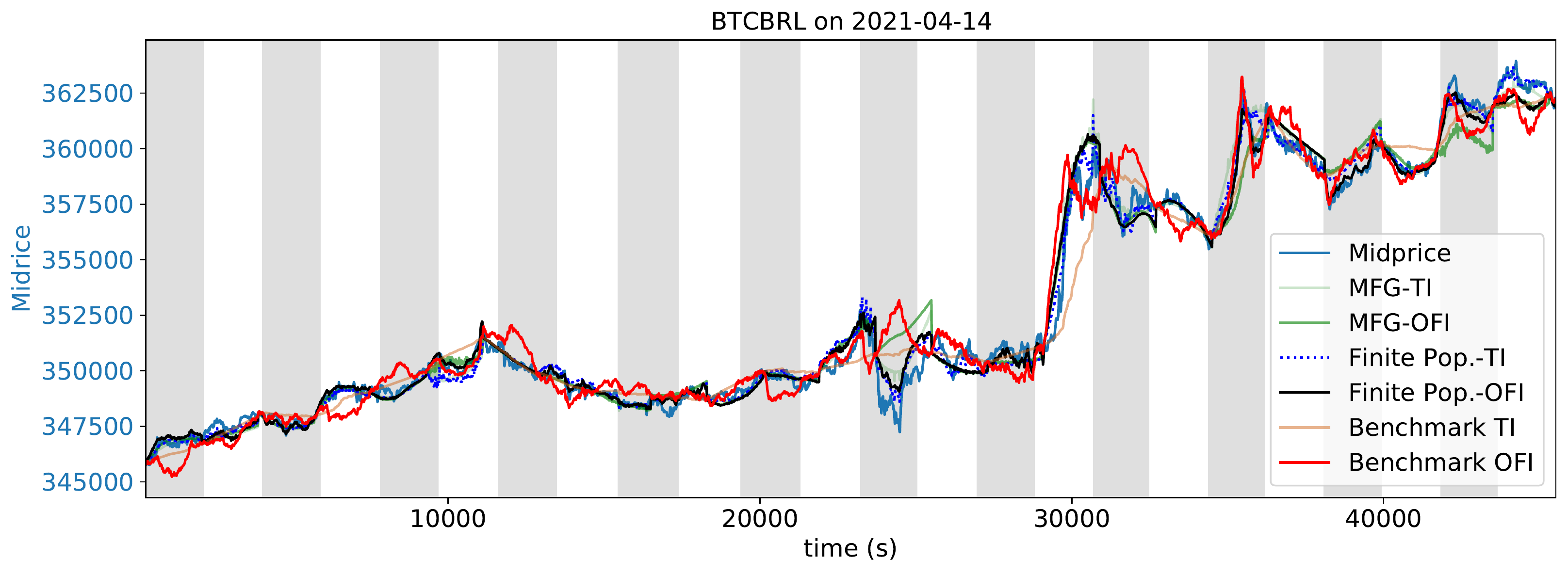}
\caption{Comparison between the midprice and formed price models: MFG with TI, MFG with OFI, Finite population with TI, and Finite population with OFI.}
\label{fig:allModelsBTCBRL}
\end{figure}

\begin{table}
\centering
\caption{Average adjusted $R^2$ results}
\begin{tabular}{lcccccc}
\toprule
  Stock &  $R^2_{TI}$ &  $R^2_{OFI}$ &  $R^2_{TI_{micro}}$ &  $R^2_{OFI_{micro}}$ &  $R^2_{TI_{bench}}$ &  $R^2_{OFI_{bench}}$ \\
\midrule
\rowcolor[gray]{0.9} BTCBRL &     0.62019 &     0.609347 &            0.687367 &             0.743341 &             0.02188 &             0.409336 \\
\bottomrule
\end{tabular}
\label{tab:R2-BTCBRL}
\end{table}

\begin{figure}[H]
\includegraphics[width=\textwidth]{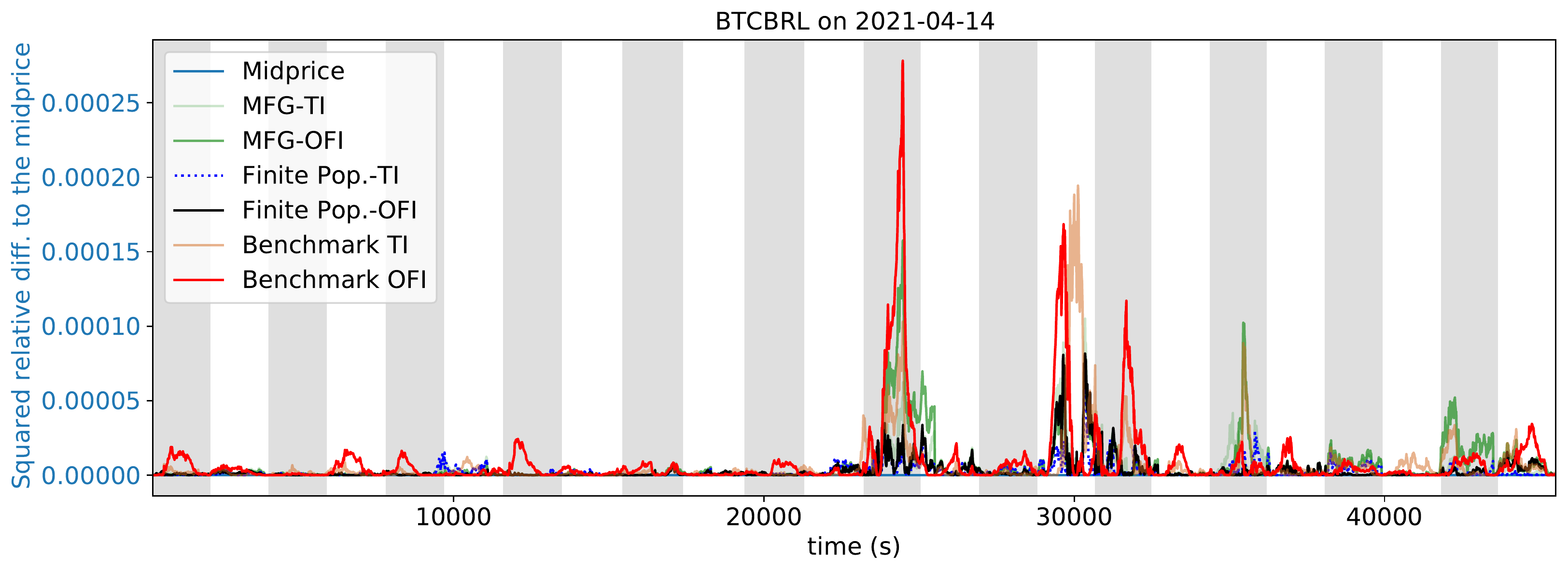}
\caption{Squared relative difference with respect to the midprice of each formed price models: MFG with TI, MFG with OFI, Finite population with TI,  and  Finite population with OFI.}
\label{fig:allModelsBTCBRL2}
\end{figure}


\end{document}